\documentclass[10pt,aps,prx,twocolumn,superscriptaddress,floatfix,noeprint]{revtex4-2}
\usepackage[utf8]{inputenc}
\usepackage[colorlinks=true,allcolors=blue]{hyperref}

\usepackage{graphicx}
\usepackage{amsmath}
\usepackage{amssymb}
\usepackage{bm}
\usepackage{mathtools}
\usepackage{braket}
\usepackage{microtype}
\usepackage{float}
\usepackage{amsthm}
\usepackage[capitalize]{cleveref}

\renewcommand{\vec}[1]{{\bm{#1}}}
\DeclarePairedDelimiter{\mean}{\langle}{\rangle}
\DeclarePairedDelimiter{\abs}{\lvert}{\rvert}

\newcommand{\di}{\mathrm{d}}
\newcommand{\comm}[2]{\left[ #1,#2 \right]}
\newcommand{\acomm}[2]{\left\{ #1,#2 \right\}}
\newcommand{\tr}{\mathrm{tr}}
\newcommand{\re}{\mathrm{Re}}
\newcommand{\norm}[1]{\left\| #1 \right\|}
\newcommand{\dagg}{^\dagger}

\newtheorem{proposition}{Proposition}

\AtBeginDocument{\addtocontents{toc}{\protect\setcounter{tocdepth}{-5}}}

\begin{document}

\title{Collective decay of interacting bosons}

\author{Bennet Windt$^\ddagger$}
\email{bennet.windt@mpq.mpg.de}
\affiliation{Max Planck Institute of Quantum Optics, 85748 Garching, Germany}
\author{Lorenzo Rossi$^\ddagger$}
\email{lorenzo.rossi@icfo.eu}
\affiliation{ICFO\textemdash Institut de Ci\`{e}ncies Fot\`{o}niques, The Barcelona Institute of Science and Technology, 08860 Castelldefels, Spain}
\author{Alexander V. Poshakinskiy}
\affiliation{ICFO\textemdash Institut de Ci\`{e}ncies Fot\`{o}niques, The Barcelona Institute of Science and Technology, 08860 Castelldefels, Spain}
\author{Daniel Malz}
\affiliation{Department of Mathematical Sciences, University of Copenhagen, 2100 Copenhagen, Denmark}
\affiliation{Department of Physics, University of Basel, 4056 Basel, Switzerland}
\author{Dominik S. Wild}
\affiliation{Max Planck Institute of Quantum Optics, 85748 Garching, Germany}
\affiliation{Google Quantum AI, 80636 Munich, Germany}

\makeatletter
\g@addto@macro\prep@absbox{%
  \vspace{-1.5em}
  \begin{center}
  \small $^\ddagger$ These authors contributed equally to this work.
  \end{center}
  \vspace{0em}%
}
\makeatother

\begin{abstract}
   We study a bosonic analog of the paradigmatic Dicke model of superradiance, comprising interacting bosonic modes subject to fully symmetric collective decay. Depending on the interaction strength, we uncover qualitatively distinct regimes of emission. For strong interactions, the emission closely resembles Dicke superradiance, with perturbative corrections arising from the presence of additional levels. For weaker interactions, the bosonic statistics qualitatively changes the dynamics, leading to a crossover to subradiant emission. Remarkably, we show that the dynamics in this regime can be described by rate equations analogous to those of the Dicke model despite the large accessible bosonic Hilbert space. Our findings are based on a combination of analytical arguments and large-scale numerics enabled by the permutational symmetry of the problem and may be probed in circuit QED experiments.
\end{abstract}

\maketitle

Ensembles of excited quantum emitters coupled to a common electromagnetic environment exhibit correlated emission arising from interference in the radiated field. As first noted by Dicke~\cite{dicke_coherence_1954,gross_superradiance_1982}, accounting for these correlations is essential to understand paradigmatic quantum optical many-body phenomena such as super- and subradiance. Following Dicke, most theoretical~\cite{agarwal_master-equation_1970,degiorgio_statistical_1970,degiorgio_approximate_1971,rehler_superradiance_1971,haake_quantum_1972,narducci_exact_1974,bonifacio_cooperative_1975,glauber_superradiant_1976,robicheaux_theoretical_2021,malz_large-n_2022,masson_universality_2022,rubies-bigorda_superradiance_2022,sierra_dicke_2022,cardenas-lopez_many-body_2023,holzinger_superradiant_2025} and experimental~\cite{skribanowitz_observation_1973,gross_observation_1976,vrehen_quantum_1977,mlynek_observation_2014,goban_superradiance_2015,kim_super-radiant_2018,ferioli_laser-driven_2021,koong_coherence_2022,liedl_observation_2023,douglas_many-body_2026} studies have focused on the cooperative emission from two-level atoms, with limited attention given more recently to the richer phenomenology resulting from multi-level spectra~\cite{lin_superradiance_2012,sutherland_superradiance_2017,pineiro_orioli_emergent_2022,masson_dicke_2024,sundar_squeezing_2024}. However, experimental advances in synthetic quantum emitters capable of hosting multiple excitations, such as superconducting circuits coupled to waveguides~\cite{lalumiere_input-output_2013,Zanner2022,orell_collective_2022, du_programmable_2026}, also open the door to exploring the correlated dissipative dynamics of emitters with bosonic character. This raises fundamental questions on the role of particle statistics in cooperative emission~\cite{marmugi_coherent_2018,jones_superradiant_2025,lin_can_2025,lu_fundamental_2025}, including whether intuition and methods developed for saturable emitters extend to bosonic systems, or whether bosonic statistics enable novel many-body dynamical regimes of correlated decay.

\begin{figure}[t!]
    \centering
    \includegraphics[width=\linewidth]{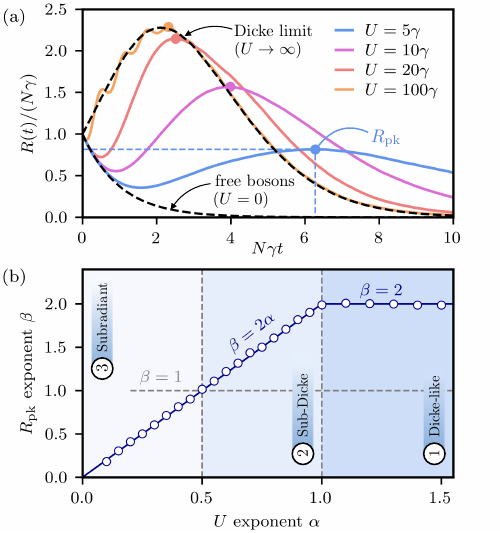}
    \caption{\emph{Bosonic collective decay.} (a) Emission rate dynamics for $N=10$ resonators subject to varied interaction strengths $U$. The limiting cases of two-level emitters ($U\to\infty$) and free bosons ($U=0$) are indicated by black dashed lines. (b) Regimes of collective emission for a tunable non-linearity $U\sim \gamma N^\alpha$, characterized by distinct asymptotic scalings of the peak emission rate $R_\mathrm{pk}\sim \gamma N^\beta$. We plot the relation between the scaling exponents $\alpha$ and $\beta$ as given by our theoretical prediction (solid line) and extracted from permutation-invariant numerical simulations (markers).}
    \label{fig:setup}
\end{figure}

In this Letter, we address these questions by studying a minimal model of anharmonic resonators realizing a bosonic analog of the Dicke master equation. We show how strong non-linearities recover the familiar Dicke superradiance, whereas weak anharmonicities give rise to markedly different, subradiant dynamics with suppressed peak emission (see Fig.~\ref{fig:setup}a). Remarkably, we find that the latter slow correlated decay is nonetheless accurately captured by rate equations on a reduced subset of states resembling a Dicke ladder. Moreover, permutation symmetry allows us to efficiently compute the associated transition rates and thus fully characterize the peak emission in the thermodynamic limit. We combine these insights with numerical simulations to identify three regimes of correlated decay. These are characterized by the asymptotic scaling of the peak emission rate with system size and range from Dicke-like dynamics to emission bursts with anomalous sub-quadratic scaling and, ultimately, the breakdown of cooperative enhancement (see Fig.~\ref{fig:setup}b).

\emph{Model.---}
We consider a system of $N$ bosonic resonators $a_j$ ($j=1,\ldots,N$), subject to a repulsive on-site interaction $U$ and collective decay under a fully symmetric jump operator. The dynamics are governed by $(\hbar=1)$
\begin{equation}
\label{eq:model}
    \dot{\rho}(t)=-i[H,\rho(t)]+\gamma\mathcal{D}[c]\rho(t)\,,
\end{equation}
where $\mathcal{D}[x]~\boldsymbol{\cdot}=x\boldsymbol{\cdot} x^\dagger-\{x^\dagger x,\boldsymbol{\cdot}\}/2$ and where
\begin{equation}
\label{eq:hamiltonian_jump}
    H=\frac{U}{2}\sum_{j}a_j^\dagger a_j^\dagger a_j a_j\,,\quad c=\sum_ja_j\,.
\end{equation}
We study the dynamics of the collective emission rate $R(t)=-\partial_t\sum_j\mean{a_j^\dagger a_j}_t=\gamma\mean{c^\dagger c}_t$ with $\mean{\ldots}_t=\mathrm{tr}[\ldots\rho(t)]$, starting from the state $\ket{\psi_0}=\prod_ja_j^\dagger|\mathrm{vac}\rangle$ (i.e. a single excitation per resonator), consistent with the canonical fully-excited initial state of the conventional Dicke problem.

There are two tractable limiting cases of Eq.~\eqref{eq:model} which already illustrate the crucial role of interactions in shaping the collective decay (see Fig.~\ref{fig:setup}a). In the limit $U\to\infty$, the Hamiltonian~\eqref{eq:hamiltonian_jump} enforces a hard-core constraint, suppressing multiple occupancies in the resonators and reducing the model to the standard Dicke master equation~\cite{dicke_coherence_1954,gross_superradiance_1982}. The dynamics are then efficiently captured by $N+1$ coupled rate equations on the Dicke ladder of symmetrized $n$-excitation states $\ket{n}$ ($n=0,1,\ldots N$), with transitions $\ket{n}\to\ket{n-1}$ occuring at rates $W_n=\gamma n(N-n+1)$. The collective emission is dominated by the brightest states around $n\sim N/2$ ($W_n\sim\gamma N^2$), which leads to the characteristic superradiant burst with a peak emission rate scaling $R_\mathrm{pk}\sim \gamma N^2$~\cite{dicke_coherence_1954,gross_superradiance_1982}.
In the opposite limiting case $U=0$, the model becomes non-interacting and displays fundamentally different dynamics: the fully symmetric bright mode contains only a single excitation on average in the initial state, resulting in a single collective decay without an emission burst~\cite{agarwal_master-equation_1970,rehler_superradiance_1971}.
In the following, we move beyond these limiting cases to explore the intermediate regimes governed by the competition between interaction and collective dissipation. We investigate the precise mechanisms by which a local emission maximum $R_\mathrm{pk}$ emerges at $U>0$ and approaches the Dicke solution as $U$ is increased (see Fig.~\ref{fig:setup}a), focusing particularly on the asymptotic scaling of $R_\mathrm{pk}$ with $N$.

\emph{Permutation-invariant subspace.---}
Given the invariance of both Eq.~\eqref{eq:model} and the initial state $\ket{\psi_0}$ under the permutation of any two resonators, we can restrict our analysis entirely to the fully symmetric sector of the full many-body Hilbert space. As discussed in more detail in the Supplemental Material~\cite{supp}, this subspace is spanned by states $\ket{\vec{x}} = \ket{x_0, x_1, x_2, \ldots}$, where $0\leq x_n\leq N$ specifies how many resonators are occupied by exactly $n$ excitations. Beyond enabling large-scale numerical simulations, this construction of the permutation-invariant subspace also invites an illustrative re-casting of our model: Following Ref.~\cite{silva_permutational_2022}, we may view $\ket{\vec{x}}$ as a multi-mode bosonic Fock state associated with annihilation operators $b_n$ ($n=0,1,2,\ldots$), in terms of which
\begin{equation}
\label{eq:symmetric_basis}
\begin{split}
    &H=\frac{U}{2}\sum_n n(n-1)b_n^\dagger b_n\,, \\
    &c=\sum_n\sqrt{n}\,b_{n-1}^\dagger b_n \eqqcolon \sum_n\sqrt{n}\,L_n\,.
\end{split}
\end{equation}
We can therefore alternatively view our system as a lattice with virtual ``sites'' associated with the modes $b_n$, subject to a non-linear on-site potential gradient and correlated incoherent nearest-neighbor tunneling (see Fig.~\ref{fig:largeU_limit}a).

\begin{figure}[t!]
    \centering
    \includegraphics[width=\linewidth]{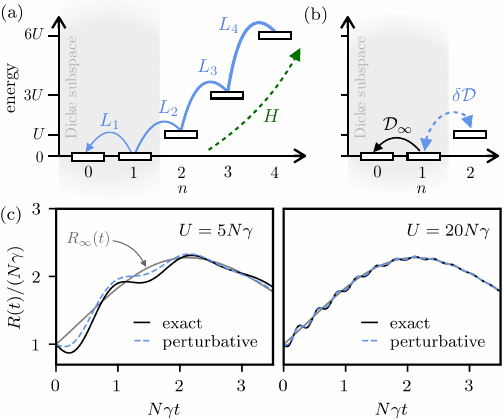}
    \caption{\emph{Large-$U$ dynamics.} (a) Schematic illustration of the auxiliary lattice model, Eq.~\eqref{eq:symmetric_basis}. The action of the Hamiltonian and the collective jump operator 
    are indicated graphically. (b) Graphical representation of the action of the diagonal dissipator $\mathcal{D}_\infty$ and the leading-order large-$U$ perturbative correction $\delta\mathcal{D}$ on the auxiliary lattice. (c) Comparison of the exact and perturbative emission dynamics for $N=10$ and two large values of $U$. The Dicke solution $R_\infty(t)$ is shown for reference.
    }
    \label{fig:largeU_limit}
\end{figure}

\emph{Large-$U$ regime.---}
Mapping Eq.~\eqref{eq:model} to the auxiliary lattice of symmetric modes~\eqref{eq:symmetric_basis} provides an advantageous starting point for analyzing the limiting regime $U\gg N\gamma$ of the dynamics. In particular, it allows us to move into the interaction picture with respect to the Hamiltonian, in which Eq.~\eqref{eq:model} takes the form $\dot{\rho}(t)=\mathcal{D}[c(t)]\rho(t)$, with $c(t)=\sum_n\sqrt{n}\,L_ne^{-iU(n-1)t}$. In the limit $U\to\infty$, this leads to fast-oscillating phases that suppress the cross-terms between jump operators $L_n,L_m$ with $n\neq m$ in $\mathcal{D}[c(t)]$. Upon dropping these terms, the Lindbladian reduces to $\mathcal{D}_\infty:=\gamma\sum_nn\mathcal{D}[L_n]$, so that the dynamics of our initial state $\ket{\psi_0}$ remain entirely in the ``Dicke subspace'' with at most one excitation per resonator (see Fig.~\ref{fig:largeU_limit}b). Consistently, the collective emission rate $R_\infty(t)=\tr[c^\dagger c\,e^{\mathcal{D}_\infty t}\rho(0)]$ corresponds exactly to the Dicke solution~\cite{supp}.

This provides a controlled starting point for understanding how deviations from Dicke scaling arise at large but finite $U$. Specifically, we can treat the cross-terms $\delta D=\mathcal{D}[c]-\mathcal{D}_\infty$ as a perturbation around the dynamics under $\mathcal{D}_\infty$, coupling the Dicke subspace to sites on the auxiliary lattice representing multi-excited sectors of the model (see Fig.~\ref{fig:largeU_limit}b). At leading order in $U^{-1}$, this results in a correction $\delta R(t)=R(t)-R_\infty(t)$ stemming from the dynamics of a virtual doublon created in the initial state. Notably, this correction can be computed from a set of $N-1$ coupled rate equations to complement the standard Dicke rate equations~\cite{supp} such that $\delta R(t)$ can be obtained with the same asymptotic complexity as the Dicke solution $R_\infty(t)$. In Fig.~\ref{fig:largeU_limit}c, we observe that this first-order perturbation theory indeed reproduces the exact emission dynamics for sufficiently large interaction strength $U$.

The rate equations for $\delta R(t)$ can also be solved approximately to lead to the illustrative expression~\cite{supp}
\begin{equation}
\label{eq:largeU_correction_approx}
    R(t)\approx\left[ 1 - 2 \frac{\gamma}{U} \sin(U t) (n_1(t) - 1) e^{-\kappa(t)} \right] R_\infty(t)\,,
\end{equation}
where $\kappa(t) = \gamma \int_0^t \di s \, (5 n_1(s) - 2 N - 6)/2$ and $n_1(t) = \tr [b_1^\dagger b_1 \, e^{\mathcal{D}_\infty t} \rho(0)]$. From Eq.~\eqref{eq:largeU_correction_approx}, we see that $R(t)$ displays initial oscillations with amplitude $N\gamma/U$ and frequency $U$ around the Dicke solution $R_\infty(t)$, which decay over a timescale $\sim(N\gamma)^{-1}$. Accordingly, if $U/(N\gamma)\to\infty$ in the limit $N\to\infty$, these corrections vanish and the emission rate $R(t)$ converges asymptotically to the Dicke solution $R_\infty(t)$. Recalling that Dicke superradiance predicts a local emission maximum at time $\sim\log N/(N\gamma)$~\cite{gross_superradiance_1982}, the dynamics inherit in particular the characteristic peak rate scaling $R_\mathrm{pk}\sim \gamma N^2$.

\emph{Small-$U$ regime.---} Having established that the predictions of Dicke superradiance are robust in the regime $U\gg N\gamma$, we now examine the opposite limiting regime $U\ll N\gamma$. To this end, it is convenient to express the system in terms of collective momentum-space modes $a_k=\sum_je^{-ikj}a_j/\sqrt{N}$. In this basis, the jump operator takes the form $c=\sqrt{N}a_0$ ($a_0\equiv a_{k=0}$), such that $R(t)=N\gamma\mean{a_0^\dagger a_0}_t$. A superradiant burst in $R(t)$ can therefore be interpreted as a macroscopic accumulation of excitations in the $k=0$ mode. Crucially, only the single $k=0$ mode decays under the collective dissipation with rate $N\gamma$, while all other $N-1$ modes are dark. Excitations in the dark subspace (i.e. the space of states $\ket{\psi}$ obeying $c\ket{\psi}=0$) can thus decay only under a second-order process, via scattering into the bright mode mediated by the Hamiltonian $H$.

To make this more explicit, the fast-decaying bright ($k=0$) mode can be adiabatically eliminated to obtain an effective master equation for the reduced state $\rho_d(t)$ of the dark ($k\neq 0$) modes~\cite{supp},
\begin{equation}
\label{eq:smallU_qme}
    \dot{\rho}_d(t)=-i[H_d,\rho_d(t)]+\sum_{\alpha=1,2}\Gamma_\alpha\mathcal{D}[c_\alpha]\rho_d(t)\,.
\end{equation}
Here, the effective Hamiltonian $H_d$ is simply the original Hamiltonian~\eqref{eq:hamiltonian_jump} projected onto the dark subspace. The dissipative terms in Eq.~\eqref{eq:smallU_qme} describe the effective loss of excitations from the dark subspace via interaction-induced scattering into the bright mode at a rate $R_d(t)=\Gamma_1\mean{c_1^\dagger c_1}_t+\Gamma_2\mean{c_2^\dagger c_2}_t$, with rates $\Gamma_1=4U^2/(N^2\gamma)$ and $\Gamma_2=U^2/(N^3\gamma)$, and associated jump operators
\begin{equation}
\label{eq:smallU_jumps}
    c_1=\frac{1}{\sqrt{N}}\sum_{\substack{k,q\neq 0 \\ k + q\neq 0}}a_{k+q}^\dagger a_k a_q \,,\quad c_2=\sum_{k\neq 0}a_ka_{-k}\,.
\end{equation}

The effective dynamics in Eq.~\eqref{eq:smallU_qme} describe a continuous scattering-induced decay of dark excitations into the bright mode. Initially, this is accompanied by a short transient during which the population in the $k=0$ mode is radiated away under $\mathcal{D}[c]$. The total emission can be viewed as the sum of these two processes, i.e. $R(t)\approx R_0(t)+R_d(t)$ where $R_0(t)=N\gamma e^{-N\gamma t}$ denotes the non-interacting ($U=0$) solution. In Fig.~\ref{fig:smallU_limit}a, we show that this effective description indeed captures the dynamics in the regime $U\ll N\gamma$.
We observe that after the initial decay, the emission rate reaches a plateau, where the fast depletion of the symmetric mode is compensated by a steady inflow from the dark manifold. The emission rate at the plateau, indicated by horizontal lines in Fig.~\ref{fig:smallU_limit}a, corresponds to the decay of the initial state via $c_1$ and $c_2$ and is approximately given by $R_d(0) \approx (4U)^2/(N\gamma)$~\cite{supp}. The plateau is followed by a peak in the emission rate and subsequent decay at late times.

\begin{figure*}[t!]
    \centering
    \includegraphics[width=\linewidth]{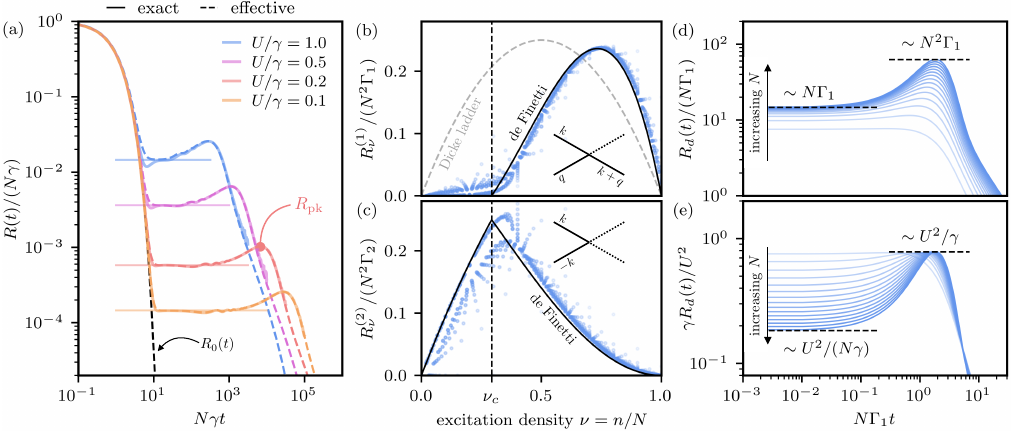}
    \caption{\emph{Small-$U$ dynamics.} (a) Comparison of the exact (solid) and effective (dashed) emission dynamics for $N=30$ and varied $U$ in the small-$U$ regime. The analytical values of the emission rate plateaus $R_d(0)$~\cite{supp} are indicated by horizontal lines and the emission rate $R_0(t)$ for $U=0$ is also included for reference. (b) Transition rates $R_n^{(1)}$ and (c) $R_n^{(2)}$ in the ground state manifold of the effective Hamiltonian $H_d$ (blue markers), for system sizes $N=10,20,\ldots,100$ (indicated by shading). The de Finetti predictions obtained in the End Matter (along with the associated critical excitation density $\nu_c=1-1/\sqrt{2}$) are also included. The scattering processes associated with the effective decay under $c_1$ and $c_2$ (see Eq.~\eqref{eq:smallU_jumps}) are indicated schematically in (b) and (c), respectively. In (b), the standard Dicke ladder transition rates are also shown for reference.
    (d), (e) Collective emission rate dynamics computed from the ground state rate equations~\eqref{eq:smallU_rate_equations} for system sizes $N=5,10,15,\ldots,80$ (indicated by shading).}
    \label{fig:smallU_limit}
\end{figure*}

To explain the emergence of the peak, it is necessary to track the evolution of the state during the plateau. The problem is significantly simplified by the separation of timescales between the coherent dynamics governed by $H_d \sim U$ and the emission: Since $\Gamma_1, \Gamma_2 \ll U$, coherences between different eigenstates of $H_d$ are effectively suppressed and the emission dynamics can be captured by rate equations in the eigenbasis of $H_d$. As shown in the Supplemental Material~\cite{supp}, these rate equations are in excellent numerical agreement with both the effective Lindbladian and the exact dynamics. While insightful, this simplification alone does not provide a dramatic computational advantage, given the exponentially large size of the eigenbasis. However, we observe numerically that the rate equation dynamics remain, to a good approximation, confined to the $(N+1)$-dimensional ground state manifold of $H_d$, spanned by the $n$-excitation ground states $\ket{\psi_n}$ ($n=0,1,\ldots, N$) of $H_d$~\cite{supp}. This observation reveals an unexpected qualitative analogy to Dicke superradiance: Despite their fundamentally different microscopic origin, the slow dynamics under Eq.~\eqref{eq:smallU_qme}, which a priori extend across the exponentially large dark subspace, in practice also reduce to coupled rate equations on an effectively one-dimensional ladder of $N+1$ states.

Explicitly restricting our attention to this set of states, the rate equations for the ground-state occupations $p_n(t)=\bra{\psi_n}\rho_d(t)\ket{\psi_n}$ take the form
\begin{equation}
\label{eq:smallU_rate_equations}
\begin{split}
    \dot{p}_n(t) &=R_{n+1}^{(1)}p_{n+1}-R_n^{(1)}p_n(t) \\
    &\qquad\quad+R_{n+2}^{(2)}p_{n+2}(t)-R_n^{(2)}p_n(t)\,,
\end{split}
\end{equation}
with transition rates $R_n^{(1)}=\Gamma_1\abs{\bra{\psi_{n-1}}c_1\ket{\psi_n}}^2$ and $R_n^{(2)}=\Gamma_2\abs{\bra{\psi_{n-2}}c_2\ket{\psi_n}}^2$. At first sight, Eq.~\eqref{eq:smallU_rate_equations} remains intractable for large systems, since it requires knowledge of the interacting many-body eigenstates $\ket{\psi_n}$ of $H_d$. Remarkably, however, the permutation invariance of the problem actually enables an efficient evaluation of the transition rates in the thermodynamic limit. We first note that to a good approximation, we can alternatively define the transition rates $R_n^{(\alpha)}$ in Eq.~\eqref{eq:smallU_rate_equations} as $R_n^{(\alpha)}=\Gamma_\alpha\bra{\psi_n}c_\alpha^\dagger c_\alpha\ket{\psi_n}$ (i.e. as the radiance from the ground state $\ket{\psi_n}$ under $c_\alpha$). In this form, the $R_n^{(\alpha)}$ depend only on the reduced two-body density operator $\rho_n^{(2)}$ associated with $\ket{\psi_n}$, which is strongly constrained by symmetry. Specifically, as discussed in more detail in the End Matter, the quantum de Finetti theorem implies that asymptotically $\rho_n^{(2)} \approx \int \di \mu(\varphi) \ket{\varphi} \! \bra{\varphi}^{\otimes 2}$ (up to corrections subleading in $N^{-1}$) under some probability measure $\mu(\varphi)$ over single-resonator states $\ket{\varphi}$~\cite{christandl_one-and--half_2007}. Because $\ket{\psi_n}$ is a dark state at excitation density $\nu=n/N$, each state in the mixture must obey $\bra{\varphi}a\ket{\varphi}=0$ and $\bra{\varphi}a^\dagger a\ket{\varphi}=\nu$. Similarly, since $\ket{\psi_n}$ is a ground state, $\ket{\phi}$ must minimize the single-particle energy $\mathcal{E}(\varphi)=\bra{\varphi}(a^\dagger)^2a^2\ket{\varphi}$.

In this manner, the evaluation of the $R_n^{(\alpha)}$, which a priori depend on nontrivial many-body states, can be formally reduced to a constrained variational problem on a single bosonic mode. This is readily carried out numerically, yielding two qualitatively different families of solutions above and below an emergent critical excitation density $\nu_c = 1-1/\sqrt{2}$ (see End Matter). The transition rates obtained in this manner are shown in Figs.~\ref{fig:smallU_limit}b,c, alongside a full numerical evaluation of the transition rates, which converges rapidly to the de Finetti prediction as $N$ is increased. Focusing on the dynamics above the critical excitation density $\nu_c$, we observe that $R_n^{(1)}$ qualitatively resembles the parabolic shape of the transition rates in the standard Dicke model, with the maximum transition rate scaling as $\sim N^2 \Gamma_1$.
Since $\Gamma_1/\Gamma_2\propto N$, transitions from $c_2$ are asymptotically suppressed at densities $\nu>\nu_c$ such that the resulting equations of motion $\dot{p}_n(t)\approx R^{(1)}_{n+1}p_{n+1}(t)-R_n^{(1)}p_n(t)$ become \emph{structurally identical} to the rate equations of the Dicke problem~\cite{dicke_coherence_1954,gross_superradiance_1982}. The resulting dynamics of $R_d(t)$ therefore exhibit the characteristic features of superradiance, with a peak emission rate $R_\mathrm{pk}\sim N^2\Gamma_1$ on a timescale $\sim (N\Gamma_1)^{-1}$ inherited from the scaling of the transition rates $R_n^{(1)}$ (see Fig.~\ref{fig:smallU_limit}d). Crucially, however, since $\Gamma_1=4U^2/(N^2\gamma)$, this quadratic enhancement is offset exactly. Despite the formal resemblance of Dicke, we therefore obtain a fundamentally different asymptotic peak emission rate scaling $R_\mathrm{pk}\sim U^2/\gamma$, displaying no explicit dependence on system size (see Fig.~\ref{fig:smallU_limit}e).

\emph{Regimes of collective emission.---} Having outlined the fundamentally different predictions for the dynamics in the limits of strong and weak interaction, we now unify these results by considering an on-site anharmonicity that scales with system size as $U\sim\gamma N^\alpha$. We are in particular interested in determining the exponent $\beta$ that governs the corresponding system-size dependence of the peak emission rate, $R_\mathrm{pk} \sim \gamma N^\beta$. 
In the limit of large $U$, our analysis indicates that perturbative corrections to the standard Dicke model are small when $\alpha > 1$, such that $\beta = 2$ in this regime. In the opposite weakly-interacting limit, our approach based on adiabatic elimination yields $R_\mathrm{pk} \sim U^2 / \gamma$ and thus $\beta = 2 \alpha$. The adiabatic elimination is justified if the occupation of the bright mode remains small. Since $\langle a_0^\dagger a_0 \rangle_t = R(t) / (N \gamma)$, this condition is self-consistently satisfied when $R_\mathrm{pk} / (N \gamma) \ll 1$ or, equivalently, $\alpha < 1/2$.

These predictions are verified in Fig.~\ref{fig:setup}b, which shows the value of $\beta$ as a function of $\alpha$ extracted from large-scale numerical simulations in the permutation-invariant subspace~\cite{supp}. The numerical results further reveal that the relation $\beta = 2 \alpha$ extends across the entire regime $0 \leq \alpha < 1$, despite the fact that it is not evident that the adiabatic elimination is valid for $\alpha \geq 1/2$. This can be reconciled by the empirical observation that the number of excitations emitted up to the peak emission time is proportional to $N$ for any $\alpha$. Therefore, while the interaction-induced flow of excitations from the dark subspace does lead to a macroscopic population $\sim N^{2\alpha-1}$ of the bright mode when $1/2 \leq \alpha < 1$, this still constitutes an asymptotically vanishing fraction $\sim N^{-2(1-\alpha)}$ of the total excitations remaining in the system at that time. This regime thus corresponds to a novel form of collective emission displaying a robust emission burst with super-linear but sub-quadratic scaling. Since the population in the bright mode cannot exceed $N$, the exponent $\beta$ saturates at $\beta = 2$, leading to the non-analyticity of $\beta$ at $\alpha = 1$ shown in Fig.~\ref{fig:setup}b, which may be viewed as a phase transition in the emission dynamics.

\emph{Conclusion \& Outlook.---}
We have demonstrated that collective decay in a bosonic system can exhibit behavior that is both closely related to and fundamentally distinct from Dicke superradiance. In particular, we have uncovered distinct phases of cooperative emission, characterized by a modified scaling of the peak emission rate and rooted in two complementary limits: In the strong-interaction regime, the dynamics asymptotically converge to those of the Dicke model, while in the weak-interaction regime, correlated emission is driven by an interaction-mediated flow of excitations out of the dark subspace, captured by a mapping to an effective Dicke ladder but resulting in slow subradiant decay. Strikingly, both regimes admit an efficient rate-equation description, indicating a shared dynamical structure underlying otherwise distinct mechanisms of collective emission.

Our results point to several intriguing directions for future exploration. For instance, considering more complex initial states with no counterpart in the two-level setting can give rise to qualitatively new dynamical features (see End Matter). Moreover, including driving terms in Eq.~\eqref{eq:model} could lead to emergent steady-state behaviour governed by the interplay between non-linearity, collective decay and coherent driving, offering a potential unified platform hosting phenomena from optical bistabilities~\cite{drummond_quantum_1980,bartolo_exact_2016,casteels_critical_2017,casteels_power_2016} to non-equilibrium quantum phase transitions~\cite{ferioli_laser-driven_2021,drummond_volterra_1978,puri_exact_1979,hassan_non-resonant_1982,kirton_introduction_2019,somech_heisenberg-langevin_2023,garraway_dicke_2011,goncalves_driven-dissipative_2024,agarwal_directional_2024}. Accordingly, our work not only advances the understanding of the role of particle statistics in superradiance but establishes interacting open bosonic systems as a fertile setting for many-body quantum optics in their own right. Our work further provides a theoretical foundation for experimental exploration, since collectively coupled bosonic modes are naturally realized in circuit QED with superconducting qubits~\cite{lalumiere_input-output_2013,Zanner2022,orell_collective_2022, du_programmable_2026} and may also be engineered with laser-cooled neutral atoms~\cite{Lewenstein_JPhysB_2000} or using the toolbox of optomechanics.

\emph{Acknowledgments.---}
We are grateful to M.~Bello for insightful discussions and for introducing us to the formalism in Ref.~\cite{silva_permutational_2022}. We thank D.~E.~Chang for enlightening conversations about the quantum de Finetti theorem as well as J.~I.~Cirac and R.~Trivedi for various helpful discussions. This research was supported in part by grant NSF PHY-2309135 to the Kavli Institute for Theoretical Physics (KITP).

\begin{center}
    \medskip\textbf{END MATTER}\medskip
\end{center}

\emph{Quantum de Finetti theorem.---} The quantum version of the de Finetti theorem provides a formal framework for understanding properties of a permutation-invariant global many-body state at the level of local properties. Since its quantitative version applies to finite-dimensional local Hilbert spaces, in our bosonic model we first limit the allowed number of excitations in each resonator to $M$. The reduced $k$-body density operator $\rho^{(k)}$ of any permutation-invariant $N$-particle state $\rho$ then obeys~\cite{christandl_one-and--half_2007}
\begin{equation}
\label{eq:de_Finetti_bound}
    \norm{\rho^{(k)}-\int\ket{\varphi}\bra{\varphi}^{\otimes k}d\mu(\varphi)}_1\leq\frac{2(M+1)k}{N}
\end{equation}
for some probability measure $\mu$ on single-particle pure states $\ket{\varphi}$. In practice, we find that choosing $M=4$ accurately captures all relevant features~\cite{supp}.

\emph{Ground state manifold of $H_d$.---}
Owing to the permutation invariance of $H_d$, Eq.~\eqref{eq:de_Finetti_bound} holds in particular for the ground states of $H_d$. We denote by $\ket{\psi_\nu}$ the ground state of the effective dark state Hamiltonian $H_d$ at a fixed excitation density $\nu$ in the thermodynamic limit $N\to\infty$, where the associated reduced $k$-particle density operator takes the form $\rho_\nu^{(k)}\approx\int\ket{\varphi}\bra{\varphi}^{\otimes k}d\mu_\nu(\varphi)$. The dark state condition ($c\ket{\psi_\nu}=0$) and fixed excitation density $\nu$ translate into constraints on the $\ket{\varphi}$, such that the measure $\mu_\nu(\varphi)$ is supported on the
feasible set
\begin{equation*}
    S_\nu=\{\ket{\varphi}: \bra{\varphi}a\ket{\varphi}=0,\;\bra{\varphi}a\dagg a\ket{\varphi}=\nu\}.
\end{equation*}
We can re-express the effective dark-state Hamiltonian as $H_d=U\sum_j\tilde{a}_j^\dagger\tilde{a}_j^\dagger\tilde{a}_j\tilde{a}_j/2$ in terms of the transformed modes $\tilde{a}_j=a_j-c/N$. Noting that for any dark state, we can simply replace $\tilde{a}_j\to a_j$ in expectation values of normal-ordered products of these transformed operators, the associated energy $E_\nu=\bra{\psi_\nu}H_d\ket{\psi_\nu}$ takes the form
\begin{equation*}
    E_\nu
    =\frac{NU}{2}\mathrm{tr}\left\{(a^\dagger)^2a^2\rho_\nu^{(1)}\right\}
    \approx\frac{NU}{2}\int_{S_\nu}\mathcal{E}(\varphi)\,d\mu_\nu(\varphi)
\end{equation*}
with $\mathcal{E}(\varphi)=\bra{\varphi}(a^\dagger)^2a^2\ket{\varphi}$, where we have dropped terms sub-leading in $N$. Analogously, since $c_1=\sum_j\tilde{a}_j^\dagger\tilde{a}_j\tilde{a}_j$ and $c_2=\sum_j\tilde{a}_j\tilde{a}_j$ (see Eq.~\eqref{eq:smallU_jumps}), the dark-state emission rates $R_\nu^{(\alpha)}=\Gamma_\alpha\bra{\psi_\nu}c_\alpha^\dagger c_\alpha\ket{\psi_\nu}$ take the form
\begin{subequations}
\label{eq:de_Finetti_rates_v1}
\begin{align}
\begin{split}
    R_\nu^{(1)}
    &\approx N^2\Gamma_1\int_{S_\nu}\abs{\bra{\varphi}a^\dagger a^2\ket{\varphi}}^2 d\mu_\nu(\varphi)\,,
\end{split} \\
\begin{split}
    R_\nu^{(2)}
    &\approx N^2\Gamma_2\int_{S_\nu}\abs{\bra{\varphi}a^2\ket{\varphi}}^2 d\mu_\nu(\varphi)\,.
\end{split}
\end{align}
\end{subequations}
In order for $\ket{\psi_\nu}$ to be a ground state, it must minimize $E_\nu$, placing a further constraint on the states $\ket{\varphi}\in S_\nu$. Specifically, since $E_\nu$ is linear in the weights of the mixture over $S_\nu$, minimizing it simply concentrates all weight on the state that has the lowest energy, which we denote by $\ket{\phi_\nu}$. Importantly, $\ket{\phi_\nu}$ is unique only up to a $U(1)$ rotation; defining $\ket{\phi_\nu(\theta)}=e^{-i\theta a^\dagger a}\ket{\phi_\nu}$, we therefore find
\begin{equation}
\label{eq:de_Finetti_state}
    \rho_\nu^{(k)}=\int_0^{2\pi}\frac{d\theta}{2\pi}P(\theta)\ket{\phi_\nu(\theta)}\bra{\phi_\nu(\theta)}^{\otimes k}\,,
\end{equation}
with some probability distribution $P(\theta)$. Hence, the expressions for the transition rates in Eqs.~\eqref{eq:de_Finetti_rates_v1} become
\begin{subequations}
\label{eq:de_Finetti_rates_v2}
\begin{align}
\begin{split}
    R_\nu^{(1)}
    &=N^2\Gamma_1\int_0^{2\pi}\frac{d\theta}{2\pi}P(\theta)\abs{\bra{\phi_\nu(\theta)}a^\dagger a^2\ket{\phi_\nu(\theta)}}^2 \\
    &=N^2\Gamma_1\abs{\bra{\phi_\nu}a^\dagger a^2\ket{\phi_\nu}}^2\,,
\end{split} \\
\begin{split}
    R_\nu^{(2)}
    &=N^2\Gamma_2\int_0^{2\pi}\frac{d\theta}{2\pi}P(\theta)\abs{\bra{\phi_\nu(\theta)}a^2\ket{\phi_\nu(\theta)}}^2 \\
    &=N^2\Gamma_2\abs{\bra{\phi_\nu} a^2\ket{\phi_\nu}}^2\,.
\end{split}
\end{align}
\end{subequations}
To determine a representative minimizer $\ket{\phi_\nu}$, we parametrize $\ket{\phi}=\sum_{n=0}^M w_n\ket{n}$ in the single-particle Fock basis, so that $\mathcal{E}(\vec{w})=\sum_n n(n-1)\abs{w_n}^2$. We then perform a multi-variate minimization of $\mathcal{E}(\vec{w})$, subject to the constraints $\sum_n\abs{w_n}^2=1$ (normalization), $\sum_n\sqrt{n+1}\,w_n^*w_{n+1}=0$ (dark state) and $\sum_n n\abs{w_n}^2=\nu$ (excitation density). Substituting the result in Eqs.~\eqref{eq:de_Finetti_rates_v2} leads to the transition rates displayed in Figs.~\ref{fig:smallU_limit}b,c. One family of low-energy dark states is given by
\begin{equation}
\label{eq:even_de_Finetti}
    \ket{\phi_\nu^\mathrm{ev}}=\sqrt{1-\frac{\nu}{2}}\ket{0}+\sqrt{\frac{\nu}{2}}\ket{2}
\end{equation}
These states have energy $\mathcal{E}=\nu$ and are provably optimal at excitation densities $\nu<\nu_c$, with critical density $\nu_c=1-1/\sqrt{2}$~\cite{supp}. In conjunction with Eqs.~\eqref{eq:de_Finetti_rates_v2},  Eq.~\eqref{eq:even_de_Finetti} implies that at $\nu<\nu_c$, the transition rates $R_\nu^{(1)}=0$ and $R_\nu^{(2)}=\nu(1-\nu/2)N^2\Gamma_2$ (up to finite-size corrections), as can be seen from Figs.~\ref{fig:smallU_limit}b,c.

\emph{Multi-excited initial states.---}
A natural generalization of our initial state is given by Fock states $\ket{\psi_0}=\prod_j(a_j^\dagger)^{M_0}\ket{\mathrm{vac}}$ with $M_0>1$ initial excitations per resonator. Practically, these states provide a significant challenge for numerical simulations, since the increased number of excitations in the system requires more generous cutoffs. However, we can still gain some insights in the limit $U\to\infty$, where we can once again consider the dynamics under the dissipator $\mathcal{D}_\infty$. Since $\mathcal{D}_\infty$ does not generate any coherences in the symmetric basis states $\ket{\vec{x}}$, the time-evolved state under $\mathcal{D}_\infty$ retains the diagonal structure $\rho(t)=\sum_\vec{x}p(\vec{x},t)\ket{\vec{x}}\bra{\vec{x}}$. The dynamics are then captured by rate equations for the probabilities $p(\vec{x},t)$ which form an $M_0$-dimensional grid of intersecting Dicke ladders associated with each of the $M_0$ jump operators $L_n$~\cite{supp}.

In Fig.~\ref{fig:multi_peak_dynamics}, we show that the resulting emission rate dynamics display $M_0$ emission peaks, all of which inherit the characteristic $\sim N^2\gamma$ scaling. This can be understood by noting that the obtained rate equations are formally equivalent to those for an ensemble of $(M_0+1)$-level systems supporting Dicke-like superradiant decay on each transition~\cite{masson_dicke_2024}.

\begin{figure}[H]
    \centering
    \includegraphics[width=\linewidth]{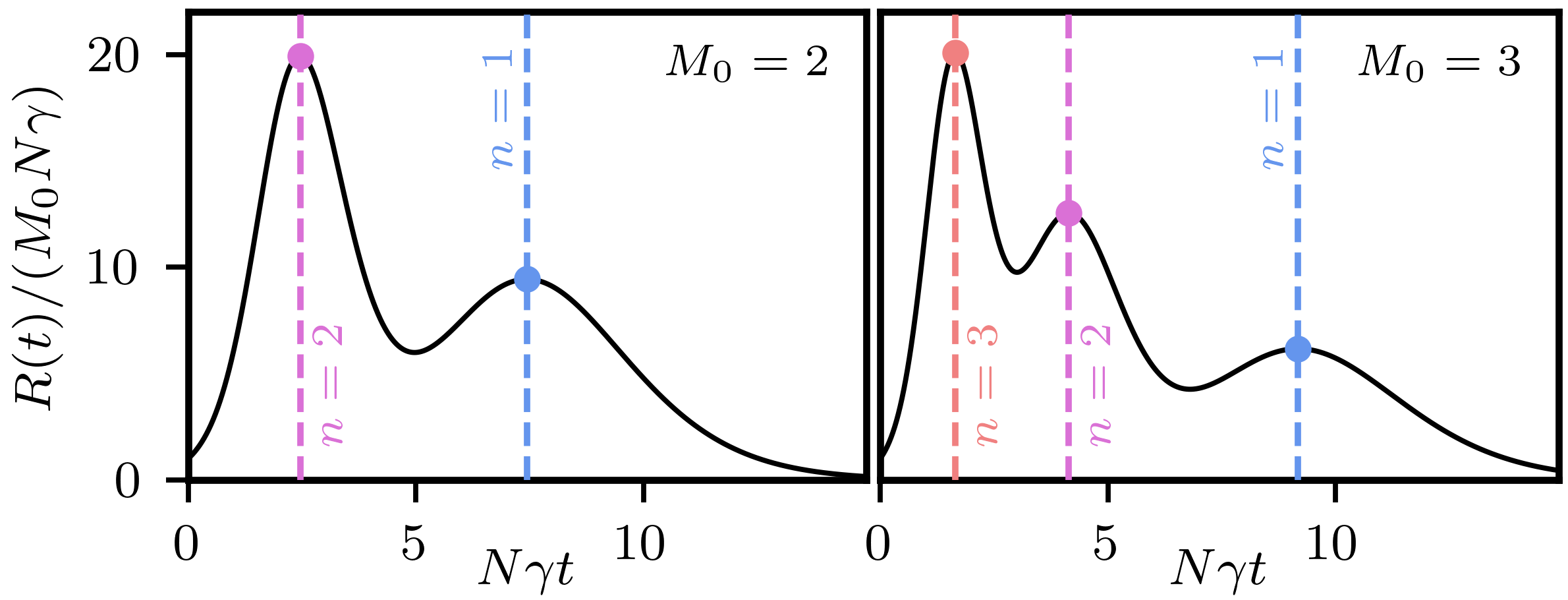}
    \caption{\emph{Multi-peak dynamics.} Dynamics of multi-excited initial Fock states with $M_0=2$ (left) and $M_0=3$ (right) initial excitations per mode, for $N=100$ modes. The emission burst peaks are indicated by vertical dashed lines.}
    \label{fig:multi_peak_dynamics}
\end{figure}

\makeatletter
\let\origaddcontentsline\addcontentsline
\let\addcontentsline\@gobblethree
\bibliography{references}
\let\addcontentsline\origaddcontentsline
\makeatother


\clearpage
\onecolumngrid

\setcounter{secnumdepth}{4}
\setcounter{section}{0}
\setcounter{figure}{0}
\renewcommand\thesection{S\arabic{section}}
\renewcommand\thefigure{S\arabic{figure}}
\setlength{\belowcaptionskip}{0pt}

\begin{center}
  \medskip
  {\large\bfseries Supplemental Material: \\ Collective decay of interacting bosons}
  \medskip
\end{center}

\makeatletter
\def\l@subsection#1#2{}
\def\l@subsubsection#1#2{}
\def\l@paragraph#1#2{}
\makeatother
\addtocontents{toc}{\protect\setcounter{tocdepth}{0}}
\tableofcontents

\section{Structure of the symmetric subspace}
\label{app:symmetric_space}

Here, we discuss in more detail the structure of the fully permutation-invariant sector as introduced in the main text, focusing in particular on the reduced dimension of the symmetric subspace as compared to the full many-body Hilbert space of the bosonic system.

\subsection{Many-body Hilbert space}
For $N$ bosonic resonators $a_j$ ($j=1,2,\ldots, N$), the full many-body Hilbert space is spanned by Fock states $\ket{n_1,n_2,\ldots n_N}$, with $n_j=0,1,2,\ldots$ denoting the number of excitations in the $j$-th mode. There are two complementary ways in which we may truncate this (in principle infinite-dimensional) space. Firstly, we can truncate the \emph{local} number of excitations in each mode at some value $M$ by restricting $n_j\leq M$. Secondly, we can truncate the \emph{global} number of excitations across all modes at some value $K$ by enforcing $\sum_j n_j\leq K$. Under these assumptions, the dimension of the Hilbert space, when viewed as the composition of $k$-excitation sectors with $k=0,1,\ldots, K$, is
\begin{equation}
\label{eq:physical_hilbert_space_dimension}
    D_\mathrm{phys}(N, M, K)=\sum_{k=0}^K\#\left\{\vec{n}=(n_1,\ldots,n_N)\in\mathbb{Z}_{\geq 0}^N~:~n_j\leq M,~\sum_jn_j=k\right\}
\end{equation}
While a general closed-form expression for $D_\mathrm{phys}(N, M, K)$ can be obtained from this counting problem, it is more illustrative to consider some limiting cases. For instance, when $K\geq NM$ (i.e. the global constraint is superfluous), we find $D_\mathrm{phys}(N,M,K\geq NM)=(M+1)^N$. Conversely, when $K\leq M$ (i.e. the local constraint is superfluous),
\begin{equation*}
    D_\mathrm{phys}(N, M, K\leq M)=\binom{K+N}{N}
\end{equation*}

\subsection{Fully-symmetric subspace}
The permutation-invariant subspace is spanned by symmetrised states labelled uniquely by their number of singlons, doublons, triplons, and so on. The basis states can therefore be denoted as $\ket{\vec{x}}=\ket{x_0,x_1,\ldots}$ where $x_n=0,1,2,\ldots,N$ denotes the number of physical modes with exactly $n$ excitations, and where $\sum_n x_n=N$ constrains the total number of modes. The \emph{local} truncation of excitations then amounts to truncating $\bm{x}$ such that $x_n=0$ for $n>M$. The \emph{global} number of excitations in $\ket{\vec{x}}$ is given by $\sum_n nx_n$, such that the space of fully-symmetric states with exactly $k$ excitations on $N$ resonators with a maximum of $M$ excitations each is given by
\begin{equation}
\label{eq:k-excitation_symmetric_subspace}
    \mathcal{H}^k(N,M)=\left\{\ket{\vec{x}}=\ket{x_0,x_1,\ldots,x_M}~:~\sum_nx_n=N,~\sum_nnx_n=k\right\}\,.
\end{equation}
It turns out that the problem of finding all states in Eq.~\eqref{eq:k-excitation_symmetric_subspace} is exactly equivalent to finding all integer partitions of $k$ into exactly $N$ parts, with each part of size $M$ at most~\cite{Andrews_2004}. Indeed, a helpful way of visualising the states in the symmetric sector is in terms of the multiplicity notation for integer partitions~\cite{Andrews_2004}, specifically as $0^{x_0}1^{x_1}2^{x_2}\ldots M^{x_M}$, where e.g. $0^N$ denotes the vacuum state, $0^{N-3}1^22^1$ denotes the symmetrised state with a two singlons and one doublon, and so on (see Fig.~\ref{fig:truncation}a). Denoting the number of integer partitions of $k$ into $N$ parts of maximum size $M$ by $p(N,M,k)$, the dimension of the full symmetric space analogous to Eq.~\eqref{eq:physical_hilbert_space_dimension} is then given by
\begin{equation}
\label{eq:symmetric_hilbert_space_dimension}
    D_\mathrm{symm}(N, M, K)=\sum_{k=0}^K\mathrm{dim}(\mathcal{H}^k(N,M))=\sum_{k=0}^K p(N,M,k)\,.
\end{equation}
While there neither exists a closed-form expression for $p(N,M,k)$ nor for the corresponding partitions $\bm{x}$, they can be obtained using dynamic-programming techniques~\cite{Sniedovich_2011}, thereby allowing us to numerically construct the symmetric basis in an efficient manner.

\begin{figure}
    \centering
    \includegraphics[width=\linewidth]{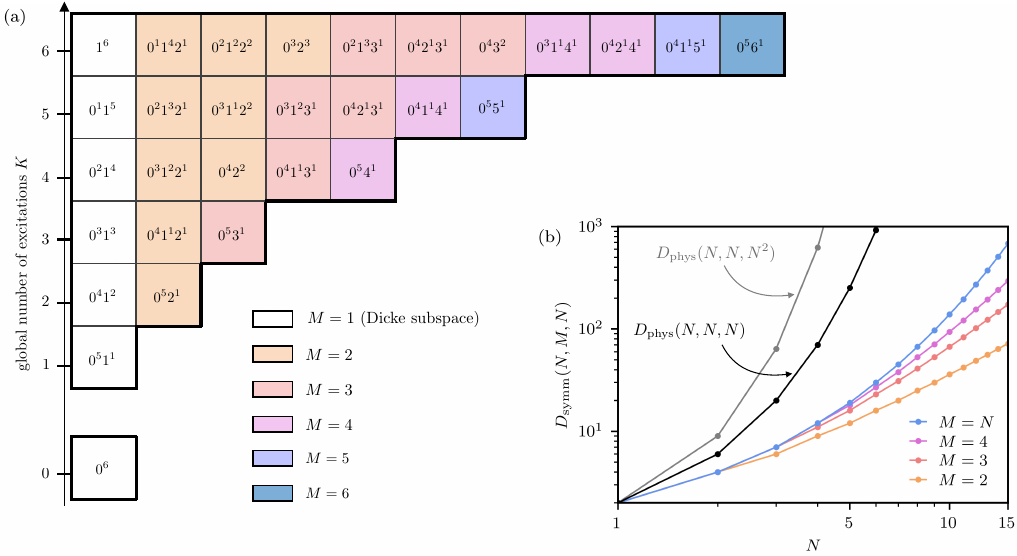}
    \caption{\emph{Permutation-invariant subspace.} (a) Visualisation of the symmetric basis for $N=6$ resonators with up to $K=6$ global excitations in terms of the multiplicity notation for integer partitions. The subsets of basis states corresponding with various local excitation number cutoffs $M\leq K$ are highlighted. (b) Scaling of the dimension~\eqref{eq:symmetric_hilbert_space_dimension} of the symmetric subspace with system size $N$ for $K=N$ and various local truncations $M<N$. The scalings of the full Hilbert space dimension~\eqref{eq:physical_hilbert_space_dimension} for $M=N$ and $K=NM=N^2$ or $K=N$ (i.e. for truncation schemes which are exact for the initial state with a single excitation per resonator) are also included for reference.}
    \label{fig:truncation}
\end{figure}

For the initial state $\ket{\psi_0}=\prod_ja_j^\dagger\ket{\rm vac}$ with a single excitation per physical mode, choosing $M=N$ and $K\geq N$ encompasses the full Hilbert space involved in the dynamics, since the master equation we consider does not increase the global number of excitations. Already at this level, the dimension of the permutation-invariant subspace~\eqref{eq:symmetric_hilbert_space_dimension} is significantly smaller than that of the full bosonic Hilbert space~\eqref{eq:physical_hilbert_space_dimension} (see Fig.~\ref{fig:truncation}b). Note that, contrary to the case of two-level emitters, the dimension of the permutation-invariant subspace still scales exponentially in $\sqrt{N}$.

Naturally, for the purpose of efficient numerical simulations, the dimension of the symmetric subspace can be further reduced by choosing a local truncation $M<N$ (see Fig.~\ref{fig:truncation}b), at the cost of incurring some potential numerical inaccuracies. However, by carefully checking the convergence of our numerics for varied $M$, we find that choosing $M=4$ yields accurate numerical results for the dynamics of the initial state $\ket{\psi_0}$ for any choice of $N$ and $U$. This reflects the fact that an accumulation of many excitations in a single resonator remains exceedingly improbable throughout the dynamics given the initial equal distribution of excitations across all resonators.

\section{Large-$U$ perturbation theory}
\label{app:largeU_perturbation}

\subsection{Perturbative correction to the emission rate}
As discussed in the main text, to deal with the limit of strong interaction, we work in the symmetric basis and split up the master equation as $\dot{\rho}=\mathcal{V}\rho\,+\mathcal{D}_\infty\rho\,+\delta \mathcal{D}\rho$, where the individual terms are given by $\mathcal{V}\rho =-i\comm{H}{\rho}$ and
\begin{subequations}
\label{eq:largeU_lindbladian_decomposition}
\begin{align}
    &\mathcal{D}_\infty\rho = \gamma \sum_{n} n \left( L_n\rho L_n^\dagger - \frac{1}{2} \acomm{L_n^\dagger L_n}{\rho} \right)=\gamma\sum_n n\mathcal{D}[L_n]\rho \label{eq:diagonal_dissipator}\\
    &\delta \mathcal{D}\rho = \gamma \sum_{n \neq m} \sqrt{nm} \left(L_m\rho L_n^\dagger- \frac{1}{2} \acomm{L_n^\dagger L_m}{\rho} \right)
\end{align}
\end{subequations}
with $L_n=b_{n-1}^\dagger b_n$. This decomposition is chosen such that $\delta \mathcal{D}$ is the only term that creates coherences between states with different occupation numbers.  We consider the effect of the term $\delta\mathcal{D}$ on the dynamics under $\mathcal{D}_\infty+\mathcal{V}$ for the initial state with a single excitation per resonator, starting from the perturbative expansion
\begin{align*}
    \rho(t) = e^{(\mathcal{D}_\infty + \mathcal{V}) t} \rho(0) 
    &+ \int_0^t \di s_1 \, e^{(\mathcal{D}_\infty + \mathcal{V}) (t - s_1)} \, \delta \mathcal{D} \, e^{(\mathcal{D}_\infty + \mathcal{V}) s_1} \rho(0) \nonumber + \ldots
\end{align*}
Keeping only terms up to first order, the decay rate $R(t)=\mean{c^\dagger c}_t$ is then given by
\begin{equation}
\label{eq:largeU_derivation_1}
    R(t) \approx \gamma \, \tr \left[ c^\dagger c \, e^{\mathcal{D}_\infty t} \rho(0) \right] + \gamma \int_0^t \di s \, \tr \left[ c^\dagger c \, e^{(\mathcal{D}_\infty + \mathcal{V}) (t - s)} \, \delta \mathcal{D} \, e^{\mathcal{D}_\infty s} \rho(0) \right] ,
\end{equation}
where we have used the fact that $\comm{\mathcal{D}_\infty}{\mathcal{V}} = 0$ and that the initial state satisfies $\mathcal{V}\rho(0) = 0$. We recognize the first term on the right-hand side as the emission rate according to conventional Dicke superradiance, denoted by $R_\infty(t)$. Since $e^{\mathcal{D}_\infty s} \rho(0)$ does not involve any multiply occupied sites,
\begin{align*}
    \delta \mathcal{D} \, e^{\mathcal{D}_\infty s} \rho(0)
    &=-\frac{\gamma}{2}\sum_{n \neq 1} \sqrt{n}\left(L_n^\dagger L_1\left(e^{\mathcal{D}_\infty s} \rho(0)\right)+\left(e^{\mathcal{D}_\infty s} \rho(0)\right)L_1^\dagger L_n\right) \\
    &=-\frac{\gamma}{\sqrt{2}}\left(L_2^\dagger L_1\left(e^{\mathcal{D}_\infty s} \rho(0)\right)+\left(e^{\mathcal{D}_\infty s} \rho(0)\right)L_1^\dagger L_2\right)\,.
\end{align*}
The resulting state has a coherence with a state involving a doubly occupied site. Hence, after evaluating the action of $\mathcal{V}$ on this state, Eq.~\eqref{eq:largeU_derivation_1} becomes, in terms of the perturbative correction $\delta R(t)\equiv R(t)-R_\infty(t)$,
\begin{equation*}
    \delta R(t)\approx - \sqrt{2} \gamma^2 \, \re \int_0^t \di s \, e^{-i U (t-s)} \, \tr \left[ c^\dagger c \, e^{\mathcal{D}_\infty (t - s)} \, L_2^\dagger L_1\left(e^{\mathcal{D}_\infty s} \rho(0)\right)\right] .
\end{equation*}
We note that the trace evaluates to a real number. In the limit of large $U$, we can integrate by parts, keeping only the boundary term. This yields
\begin{equation*}
    \delta R(t) \approx - \frac{\sqrt{2}\gamma^2}{U} \sin(U t) \, \tr \left[ c^\dagger c \, e^{\mathcal{D}_\infty t} \, L_2^\dagger L_1  \rho(0) \right] .
\end{equation*}
For the trace to be non-vanishing, the $c^\dagger c$ term has to remove the coherence with the doubly occupied state. There is only one such term in the expansion of $c^\dagger c$ in terms of the $L_n$, thereby we obtain
\begin{equation}
    \delta R(t) \approx - \frac{2\gamma^2}{U} \sin(U t) \, \tr \left[ L_1^\dagger L_2 \, e^{\mathcal{D}_\infty t} \, L_2^\dagger L_1   \rho(0) \right] .
    \label{eq:rate_pt}
\end{equation}

\subsection{Rate equations for coherences}
Numerically, the complexity of evaluating the perturbative correction~\eqref{eq:rate_pt} is comparable to evaluating the Dicke solution $R_\infty(t)$. To see this, we recall that the Dicke emission rate can be expressed as~\cite{gross_superradiance_1982}
\begin{equation}
    R_\infty(t) = \gamma\sum_{n = 0}^N n (N-n+1) p_n(t)=\sum_{n=0}^NW_n\,p_n(t)\,,
    \label{eq:rate_dicke}
\end{equation}
where $p_n(t) = \braket{N-n, n, 0, 0, \ldots | e^{\mathcal{D}_\infty t} \rho(0) | N-n, n, 0, 0, \ldots}$ denotes the population of the Dicke state with $n$ occupations, which in the symmetric basis takes the form $\ket{N-n,n,0,0,\ldots}$~\cite{silva_permutational_2022}, and where $W_n= \gamma n (N-n+1)$. The populations $p_n(t)$ obey differential equations~\cite{gross_superradiance_1982}
\begin{equation}
    \dot{p}_n(t) = W_{n+1} p_{n+1}(t) - W_n p_n(t)\,,
    \label{eq:populations}
\end{equation}
with initial conditions $p_n(0) = \delta_{n,N}$. In fact, the correction~\eqref{eq:rate_pt} can be evaluated using similar rate equations; to construct them, we define the operator $A_n = \sqrt{(n-1) n (N - n + 1)} \ket{N-n+1, n-2, 1, 0, 0, \ldots} \bra{N-n, n, 0, 0, 0, \ldots}$ for any $2 \leq n \leq N$, corresponding to the coherence created by $L_2^\dagger L_1=b_2^\dagger b_0^\dagger b_1^2$ acting on the Dicke state with $n$ excitations, with a numerical pre-factor chosen for later convenience. The operators satisfy $\mathcal{D}_\infty A_n = W_n A_{n-1} - V_n A_n$, with additional coefficients $V_n = \gamma[(n-2)(N-n+2) + n(N-n+1) + 2(n-1)]/2$. Therefore,
\begin{equation*}
    e^{\mathcal{D}_\infty t} L_2^\dagger L_1 \rho(0) = \sum_{n=2}^N c_n(t) A_n\,,
\end{equation*}
where the coefficients $c_n(t)$ in the expansion are governed by the coupled differential equations 
\begin{equation}
    \dot{c}_n(t) = W_{n+1} c_{n+1}(t) - V_n c_n(t)\,,
\label{eq:coherences}
\end{equation}
with the initial condition $c_n(0) = \delta_{n, N}$. We can thus obtain the decay rate from Eq.~\eqref{eq:rate_pt} as
\begin{equation}
    \delta R(t) = - \frac{\sqrt{2}\gamma^2}{U} \sin(U t) \,\sum_{n = 2}^N (n-1) n (N-n+1) c_n(t) .
    \label{eq:coherences_sum}
\end{equation}
The plots comparing the perturbative solution to the exact dynamics in the main text are obtained from Eqs.~\eqref{eq:rate_dicke} and~\eqref{eq:coherences_sum}, using numerical solutions of the differential equations~\eqref{eq:populations} and~\eqref{eq:coherences}.

\subsection{Approximate solution}
As noted in the main text, beyond allowing for the efficient numerical evolution of the perturbative correction to the Dicke solution $R_\infty(t)$ in the large-$U$ regime, Eq.~\eqref{eq:coherences_sum} also provides a starting point for obtaining an approximative expression for $\delta R(t)$ which allows us to extract the modification of the peak emission rate predicted by the Dicke theory. Specifically, we observe that Eqs.~\eqref{eq:coherences} and \eqref{eq:populations} are equivalent up to the difference $V_n - W_n = (5 n - 2 N - 6) / 2$, which suggests that $c_n(t)$ may be approximately related to $p_n(t)$. 

To exploit this, we assume that the coefficients $p_n(t)$ and $c_n(t)$ are peaked around $n_1(t) = \sum_{n = 0}^N n p_n(t) = \langle b_1^\dagger b_1 \, e^{\mathcal{D}_\infty t} \rho(0) \rangle$, which is the expected number of remaining excitations at time $t$ during the Dicke evolution. This assumption allows us to simplify expressions by replacing $n$ by $n_1(t)$. In particular, we replace $V_n$ in Eq.~\eqref{eq:coherences} by the time-dependent coefficient $\widetilde{V}_n(t) \approx \alpha_n + \gamma (5 n_1(t) - 2 N - 6)/2$. Since $\widetilde{V}_n(t) - W_n$ is independent of $n$, we arrive at
\begin{equation}
    c_n(t) \approx e^{-\kappa(t)} p_n(t),
    \label{eq:approx_coherences}
\end{equation}
where $\kappa(t) = \gamma \int_0^t \di s \, (5 n_1(s) - 2 N - 6)/2$. Following the same argument, we approximate the sum on the right-hand side of Eq.~\eqref{eq:coherences_sum} according to
\begin{equation}
    \sum_{n = 2}^N (n-1) n (N-n+1) c_n(t) 
    \approx (n_1(t) - 1)  \sum_{n = 2}^N n (N-n+1) c_n(t)
    \approx \frac{e^{-\kappa(t)}}{\gamma}(n_1(t) - 1) R_\infty(t)\,,
\end{equation}
where we have used Eqs.~\eqref{eq:approx_coherences} and~\eqref{eq:rate_dicke}. Thus, we finally obtain the expression given in the main text,
\begin{equation}
    R(t) \approx \left[ 1 - 2 \frac{\gamma}{U} \sin(U t) (n_1(t) - 1) e^{-\kappa(t)} \right] R_\infty(t) .
\end{equation}

\section{Small-$U$ effective master equation \& rate equations}
\subsection{Adiabatic elimination of the bright mode}
In the small-$U$ regime, it is convenient to re-cast the dynamics in terms of the collective momentum-space modes $a_k=\sum_j e^{-ikj} a_j/\sqrt{N}$. The jump operator can be written as $c=\sqrt{N}a_{k=0}$ while the Hamiltonian reads
\begin{equation}
\label{eq:momentum_hammiltonian}
	H=\frac{U}{2N}\sum_{\{k_i\}} \delta_{\{k_i\}}a^\dagger_{k_1}a^\dagger_{k_2}a_{k_3}a_{k_4} \quad ,
\end{equation}
where we use the shorthand notation $\delta_{\{k_i\}}$ to denote imposed momentum conservation, i.e. in this case $\delta_{\{k_i\}}=\delta_{0,k_1+k_2-k_3-k_4}$. As discussed in the main text, in the small-$U$ limit we distinguish between a single bright mode ($k=0$) and $N-1$ dark modes ($k\neq 0$). An excitation in the dark manifold can decay only after scattering into the bright mode via the non-linear interaction in $H$. If this scattering process (timescale $\sim 1/U$) is much slower than the decay of the $k=0$ mode itself (timescale $\sim 1/(N\gamma)$), we can adiabatically eliminate the latter and obtain an effective description of the dynamics in the dark manifold alone. 

From a technical standpoint, it is convenient to formalize this idea by applying the generalized Schrieffer-Wolff transformation described in Ref.~\cite{kessler_generalized_2012}. To do so, we denote $a_0\equiv a_{k=0}$ and we split the Linblad super-operator as $\mathcal{L}=\mathcal{L}_0 + \mathcal{V}$, where $\mathcal{L}_0=\gamma\mathcal{D}[c]=\gamma N\mathcal{D}[a_0]$ and $\mathcal{V}~\boldsymbol{\cdot}=-i[H, \, \boldsymbol{\cdot} \,]$. We then look for an effective Linbladian $\mathcal{L}_d$ acting only on the dark subspace of $\mathcal{L}_0$ and incorporating the dynamics generated by the perturbation $\mathcal{V}$. Up to second order, this is given by \cite{kessler_generalized_2012}
\begin{equation}
\label{eq:effective_lindbladian_definition}
	\mathcal{L}_d \approx P \mathcal{V} P- P \mathcal{V} Q \mathcal{L}_0^{-1} Q \mathcal{V} P \quad ,
\end{equation}
where $P$ is the projector onto the dark subspace of $\mathcal{L}_0$, namely $\mathcal{L}_0 P=P\mathcal{L}_0=0$, while $Q=\mathbf{1}-P$.

\subsubsection{Projector onto the dark subspace}
To explicitly build the projector $P$ we need to find the right and left eigenvectors of $\mathcal{L}_0$ with vanishing eigenvalue. We label the normalized bosonic Fock states in the momentum-mode basis as $\ket{n_0,\vec{n}}$, where $n_0$ denotes the number of excitations in the $k=0$ mode, $\vec{n}$ denotes the number of excitations in the other $N-1$ modes and $\braket{n_0,\vec{n} | m_0,\vec{m}}=\delta_{n_0 m_0}\delta_{\vec{n} \vec{m}}$. A set of linearly independent right eigenvectors of $\mathcal{L}_0$ with zero eigenvalue can then be chosen as
\begin{equation*}
	R_{(\vec{n}\vec{m})} = \ket{0,\vec{n}}\bra{0,\vec{m}}
\end{equation*}
given that $\mathcal{L}_0$ only acts on the $k=0$ mode and its vacuum is the unique dark state. The left eigenvectors of $\mathcal{L}_0$ can instead be recovered by finding the corresponding right eigenvectors of the adjoint operator $\mathcal{L}_0^\dagger$. By inspection, one finds that the right eigenvectors of $\mathcal{L}_0^\dagger$ with zero eigenvalues are given by all operators acting as the identity on the $k=0$ mode. We can then take as left eigenvectors of $\mathcal{L}_0$ with zero eigenvalue the set of linearly independent operators
\begin{equation*}
	L_{(\vec{n}\vec{m})} = \sum_{n_0} \ket{n_0,\vec{n}}\bra{n_0,\vec{m}} \quad .
\end{equation*}
The chosen operators obey the biorthonormality condition $\tr\{L^\dagger_{(\vec{n}\vec{m})} R_{(\vec{r}\vec{s})}\} = \delta_{(\vec{n}\vec{m})(\vec{r}\vec{s})}$ and the projector $P$ onto the dark subspace can then finally be constructed as
\begin{equation}
    P[\,\boldsymbol{\cdot}\,]= \sum_{(\vec{n}\vec{m})} R_{(\vec{n}\vec{m})} \tr\{L^\dagger_{(\vec{n}\vec{m})} \, \boldsymbol{\cdot} \,\}= \sum_{n_0}\sum_{(\vec{n}\vec{m})}\ket{0,\vec{n}} \bra{n_0,\vec{n}}\,\boldsymbol{\cdot}\,\ket{n_0,\vec{m}}\bra{0,\vec{m}}\,.
\end{equation}

\subsubsection{First-order effective Lindbladian}
Having constructed the dark state projector $P$, we can now proceed with the derivation the effective Linbladian. The first-order term $\mathcal{L}_d^{(1)}=P\mathcal{V}P$ in the expansion~\eqref{eq:effective_lindbladian_definition} can be computed as
\begin{align*}
	\mathcal{L}_d^{(1)}\rho 
	&=-i\frac{U}{2N}\sum_{\{k_i\}} \delta_{\{k_i\}} P[a^\dagger_{k_1}a^\dagger_{k_2}a_{k_3}a_{k_4} P[\rho]] + \text{h.c.} \\
	&=-i\frac{U}{2N} \sideset{}{'}\sum_{\{k_i\}} \delta_{\{k_i\}} a^\dagger_{k_1}a^\dagger_{k_2}a_{k_3}a_{k_4} P[\rho] + \text{h.c.} 
\end{align*}
where we have introduced the notation of a primed sum to mean that we exclude the $k=0$ mode, and where we have also taken advantage of the fact that $a_0P[\,\boldsymbol{\cdot}\,]=P[\,\boldsymbol{\cdot}\,]a_0^\dagger=0$. Evidently, this can be alternatively expressed as $\mathcal{L}_d^{(1)}\,\boldsymbol{\cdot}\, = -i [H_d,P[\,\boldsymbol{\cdot}\,]]$, where we have defined the effective dark state Hamiltonian as
\begin{equation}
\label{eq:effective_hamiltonian}
H_d=\frac{U}{2N} \sideset{}{'}\sum_{\{k_i\}} \delta_{\{k_i\}} a^\dagger_{k_1}a^\dagger_{k_2}a_{k_3}a_{k_4}\,.
\end{equation}
A comparison with Eq.~\eqref{eq:momentum_hammiltonian} shows that $H_d$ is simply the full Hamiltonian $H$ restricted to the dark modes.

\subsubsection{Second-order effective Lindbladian}
To compute the second-order term $\mathcal{L}_d^{(2)}=P \mathcal{V} Q \mathcal{L}_0^{-1} Q \mathcal{V} P$, we first use the fact that $Q=\mathbf{1}-P$ to re-express $Q \mathcal{V} P[\,\boldsymbol{\cdot}\,]=\mathcal{V}P[\,\boldsymbol{\cdot}\,]-\mathcal{L}_d^{(1)}\,\boldsymbol{\cdot}\,=-i[H-H_d,P[\,\boldsymbol{\cdot}\,]]$ and hence, using Eqs.~\eqref{eq:momentum_hammiltonian} and~\eqref{eq:effective_hamiltonian},
\begin{equation*}
    Q \mathcal{V} P[\rho]=-i\frac{U}{N} a^\dagger_{0} \sideset{}{'}\sum_{\{k_i\}} \delta_{\{k_i\}} a^\dagger_{k_1}a_{k_2}a_{k_3} P[\rho]-i\frac{U}{2N} a^\dagger_{0}a^\dagger_{0} \sideset{}{'}\sum_{\{k_i\}} \delta_{\{k_i\}} a_{k_1}a_{k_2} P[\rho] + \text{h.c.} \quad .
\end{equation*}
We then observe that
\begin{align*}
	\mathcal{L}_0\Big[ a^\dagger_{0} \sideset{}{'}\sum_{\{k_i\}} \delta_{\{k_i\}} a^\dagger_{k_1}a_{k_2}a_{k_3} P[\rho] \Big] &= -\frac{N\gamma}{2} \, a^\dagger_{0} \sideset{}{'}\sum_{\{k_i\}} \delta_{\{k_i\}} a^\dagger_{k_1}a_{k_2}a_{k_3} P[\rho] \\
	\mathcal{L}_0\Big[ a^\dagger_{0}a^\dagger_{0} \sideset{}{'}\sum_{\{k_i\}} \delta_{\{k_i\}} a_{k_1}a_{k_2} P[\rho] \Big] &= -N\gamma \, a^\dagger_{0}a^\dagger_{0} \sideset{}{'}\sum_{\{k_i\}} \delta_{\{k_i\}} a_{k_1}a_{k_2} P[\rho] 
\end{align*}
and thus we obtain
\begin{align}
	Q\mathcal{L}_0^{-1}Q \mathcal{V} P[\rho] &= i\frac{2U}{N^2\gamma} a^\dagger_{0} \sideset{}{'}\sum_{\{k_i\}} \delta_{\{k_i\}} a^\dagger_{k_1}a_{k_2}a_{k_3} P[\rho] \nonumber +i\frac{U}{2N^2\gamma} a^\dagger_{0}a^\dagger_{0} \sideset{}{'}\sum_{\{k_i\}} \delta_{\{k_i\}} a_{k_1}a_{k_2} P[\rho] + \text{h.c.} \quad .
\end{align}
Finally, we write the second-order contribution as
\begin{equation}
\begin{split}
	\mathcal{L}_d^{(2)} [\rho] =& -\frac{2U^2}{N^3\gamma} \sideset{}{'}\sum_{\{p_i\}} \sideset{}{'}\sum_{\{k_i\}} \delta_{\{p_i\}} \delta_{\{k_i\}} a^\dagger_{p_1}a^\dagger_{p_2}a_{p_3} a^\dagger_{k_1}a_{k_2}a_{k_3} P[\rho] \nonumber \\
	& +\frac{2U^2}{N^3\gamma} \sideset{}{'}\sum_{\{p_i\}} \sideset{}{'}\sum_{\{k_i\}} \delta_{\{p_i\}} \delta_{\{k_i\}} a^\dagger_{k_1}a_{k_2}a_{k_3} P[\rho] a^\dagger_{p_1}a^\dagger_{p_2}a_{p_3} \nonumber \\
	& -\frac{U^2}{2N^3\gamma} \sideset{}{'}\sum_{\{p_i\}} \sideset{}{'}\sum_{\{k_i\}} \delta_{\{p_i\}} \delta_{\{k_i\}} a^\dagger_{p_1}a^\dagger_{p_2}a_{k_1} a_{k_2} P[\rho]\nonumber \\
	& +\frac{U^2}{2N^3\gamma} \sideset{}{'}\sum_{\{p_i\}} \sideset{}{'}\sum_{\{k_i\}} \delta_{\{p_i\}} \delta_{\{k_i\}} a_{k_1} a_{k_2} P[\rho] a^\dagger_{p_1}a^\dagger_{p_2} + \text{h.c.}
	\quad .
\end{split}
\end{equation}
To re-cast $\mathcal{L}_d^{(2)}$ in its Lindblad form, we can define the effective dark mode jump operators
\begin{equation}
	c_1 = \frac{1}{\sqrt{N}} \sideset{}{'}\sum_{\{k_i\}} \delta_{\{k_i\}} a^\dagger_{k_1}a_{k_2}a_{k_3}\,,\qquad
	c_2 =  \sideset{}{'}\sum_{\{k_i\}} \delta_{\{k_i\}} a_{k_1}a_{k_2}
\end{equation}
which allows us to re-write $\mathcal{L}_d^{(2)}~\boldsymbol{\cdot}\,=\Gamma_1\mathcal{D}[c_1]P[\,\boldsymbol{\cdot}\,]+\Gamma_2\mathcal{D}[c_2]P[\,\boldsymbol{\cdot}\,]$ with $\Gamma_1=4U^2/(N^2\gamma)$ and $\Gamma_2=U^2/(N^3\gamma)$. Putting everything together, we then arrive at the Lindbladian corresponding with the effective master equation given in the main text,
\begin{equation}
\label{eq:effective_lindbladian}
    \mathcal{L}_d~\boldsymbol{\cdot}\,=-i [H_d,P[\,\boldsymbol{\cdot}\,]]+\sum_{\alpha=1,2}\Gamma_\alpha\mathcal{D}[c_\alpha]P[\,\boldsymbol{\cdot}\,]
\end{equation}

\subsection{Initial state for the effective dynamics}

\subsubsection{Initial state in the dark subspace}
The canonical initial state which we choose for our dynamics can be expressed in momentum space as a superposition of states with $n=0,1,\ldots,N$ excitations in the bright mode, i.e.
\begin{equation}
\label{eq:initial_state_momentum}
    \ket{\psi_0}=\prod_ja_j^\dagger\ket{\mathrm{vac}}=\prod_j\left(\tilde{a}_j^\dagger+\frac{a_0}{\sqrt{N}}\right)\ket{\mathrm{vac}}=\sum_{n=0}^{N}C(n)\ket{N-n,\phi_n}\,,
\end{equation}
where we have introduced the modes $\tilde{a}_j=a_j-a_0/\sqrt{N}=a_j-c/N$ and where, adopting the notation from the previous section, $\ket{m,\phi_n}$ denotes the an $m$-excitation Fock state on the bright mode and an $n$-excitation state $\ket{\phi_n}$ on the dark modes, which we can explicitly express as
\begin{equation}
\label{eq:phi_and_J_states}
    \ket{\phi_n}=\frac{1}{f(n)}\sum_{|J|=n}\ket{J}\quad\text{with}\quad\ket{J}=\prod_{j\in J}\tilde{a}_j^\dagger\ket{\mathrm{vac}}\,.
\end{equation}
Note that since the modes $\tilde{a}_j$ obey the operatorial identity $\sum_j\tilde{a}_j=0$, the state $\ket{\phi_1}$ does not exist; in other words, there is trivially no single-excitation state that is dark under $c$ (see Fig.~\ref{fig:dark_subspace}). In the definition~\eqref{eq:phi_and_J_states}, $f(n)$ serves as a normalisation constant, ensuring that $\langle\phi_n|\phi_m\rangle=\delta_{nm}$ (since orthogonality follows from the fact that $\langle J|J'\rangle=0$ for $\abs{J}\neq\abs{J'}$). Motivated by the clear separation of time scales between the fast depletion of the bright mode and the subsequent dark-mode dynamics in the small-$U$ limit, we choose the initial state $\rho_0$ for the effective master equation~\eqref{eq:effective_lindbladian} as the steady state of $\mathcal{L}_0=\gamma\mathcal{D}[c]$,
\begin{equation}
\label{eq:initial_state_dark}
\rho_0=\lim_{t\to\infty}\,e^{\mathcal{L}_0t}\ket{\psi_0}\bra{\psi_0}=\sum_{n=0}^{N}\abs{C(n)}^2\ket{0,\phi_n}\bra{0,\phi_n}\,.
\end{equation}

\subsubsection{Initial distribution of excitations}
We now turn to the distribution of excitations in the dark subspace $\abs{C(n)}^2=\abs{f(n)}^2n!/N^{N-n}$, for which it is possible to obtain an analytical expression; however, this requires some technical developments. First, we note that $[\tilde{a}_i,\tilde{a}_j^\dagger]=G_{ij}$, where $G_{ij}=\delta_{ij}-1/N$. For two states $\ket{J},\ket{J'}$ as defined in Eq.~\eqref{eq:phi_and_J_states} with $J=(j_1,j_2,\ldots,j_n)$ and $J'=(j'_1,j'_2,\ldots,j'_n)$, a straightforward application of Wick's Theorem then yields $\langle J|J'\rangle=\mathrm{perm}(\bm{G}_{JJ'})$, where $\bm{G}_{JJ'}$ denotes the restriction of $\bm{G}$ to elements $G_{jj'}$ with $j\in J$ and $j'\in J'$. In fact, the permanent of $\bm{G}$ (and its restrictions) can be straightforwardly computed analytically, whereby we obtain
\begin{equation}
\label{eq:J_overlaps}
    \langle J|J'\rangle=\sum_{r=0}^t\binom{t}{r}(n-r)!\left(-\frac{1}{N}\right)^{n-r}=: P(n,t)\quad\text{for}\quad\abs{J}=\abs{J'}=n,~\abs{J\cap J'}=t
\end{equation}
In turn, the property~\eqref{eq:J_overlaps} now allows us to compute the normalisation constant $f(n)$ according to
\begin{align*}
    \abs{f(n)}^2
    &=\sum_{\abs{J}=n}\sum_{\abs{J'}=n}\langle J|J'\rangle \\
    &=\sum_{\abs{J}=n}\sum_{t=0}^n P(n,t)~\#\{J':\abs{J'}=n,~\abs{J\cap J'}=t\} \\
    &=\binom{N}{n}\sum_{t=0}^n\binom{n}{t}\binom{N-n}{n-t}\sum_{r=0}^t\binom{t}{r}(n-r)!\left(-\frac{1}{N}\right)^{n-r} \\
    &=\binom{N}{n}\sum_{k=0}^n\,k!\binom{n}{k}\binom{N+k-n}{k}\left(-\frac{1}{N}\right)^k
\end{align*}
where we have used Vandermonde's identity and omitted some lengthy but straightforward algebra. The initial distribution of excitations $\abs{C(n)}^2$ is therefore given by
\begin{equation}
\label{eq:initial_state_dark_distribution}
    \abs{C(n)}^2=\frac{n!}{N^{N-n}}\binom{N}{n}\sum_{k=0}^n\,k!\binom{n}{k}\binom{N+k-n}{k}\left(-\frac{1}{N}\right)^k
\end{equation}

\subsection{Analytical emission rate plateau}
While the state~\eqref{eq:initial_state_dark} serves as the initial state for the effective small-$U$ dynamics, the emission rate plateau value $R_d(0)$ can also be calculated directly from the state~\eqref{eq:initial_state_momentum}, noting that $c_1,c_2$ leave the bright mode invariant and hence
\begin{align*}
    \bra{\psi_0}c_\alpha^\dagger c_\alpha\ket{\psi_0}=\sum_{n,m}C^*(n)C(m)\underbrace{\langle n|m\rangle}_{\textstyle\delta_{nm}}\bra{\phi_n}c_\alpha^\dagger c_\alpha\ket{\phi_m}=\sum_n\abs{C(n)}^2\bra{\phi_n}c_\alpha^\dagger c_\alpha\ket{\phi_n}=\mathrm{tr}\{c_\alpha^\dagger c_\alpha \rho_0\}\,.
\end{align*}
Accordingly, $R_d(0)=\sum_\alpha\Gamma_\alpha\bra{\psi_0}c_\alpha^\dagger c_\alpha\ket{\psi_0}$ and this expression allows for a more straightforward evaluation in the symmetric basis. To this end, we first re-express the effective jump operators in this basis as
\begin{subequations}
\label{eq:c1_c2_symmetric_basis}
\begin{align}
    &c_1=\sum_n\sqrt{n}(n-1)\,b_{n-1}^\dagger b_n-\frac{1}{N}\sum_n\sqrt{n(n-1)}\,c^\dagger b_{n-2}^\dagger b_n-\frac{2}{N}\sum_nn\,b_n^\dagger b_nc+\frac{2}{N^2}\,c^\dagger c^2 \\
    &c_2=\sum_n\sqrt{n(n-1)}\,b_{n-2}^\dagger b_n-\frac{1}{N}\,c^2
\end{align}
\end{subequations}
We now focus on the initial state with $M_0=1$ (i.e. a single initial excitation per physical mode) which we express in the symmetric basis as $\ket{0,N,0,\ldots}$. Using Eqs.~\eqref{eq:c1_c2_symmetric_basis}, we find that
\begin{align*}
    &c_1\ket{0,N,0,0,\ldots}=-\frac{2(N-1)(N-2)}{N\sqrt{N}}\ket{1,N-1,0,0,\ldots}+\frac{4\sqrt{(N-1)(N-2)}}{N\sqrt{N}}\ket{2,N-3,1,0,\ldots} \\
    &c_2\ket{0,N,0,0,\ldots}=-\frac{\sqrt{2(N-1)}}{\sqrt{N}}\ket{2,N-2,0,0,\ldots}
\end{align*}
Noting that all the states on the right hand side of these expressions are orthogonal to each other, we can then compute $R_d(0)=\sum_\alpha\Gamma_\alpha\mean{c_\alpha^\dagger c_\alpha}_{0}$ with $\Gamma_1=(2U)^2/(N^2\gamma)$ and $\Gamma_2=U^2/(N^3\gamma)$ to obtain
\begin{equation}
\begin{split}
    R_d(0)
    &=\frac{4(N-1)(N-2)\left[4+(N-1)(N-2)\right]\Gamma_1}{N^3}+\frac{2(N-1)\Gamma_2}{N} \\
    &=\frac{U^2}{\gamma}\left(\frac{4(N-1)(N-2)\left[4+(N-1)(N-2)\right]}{N^5}+\frac{2(N-1)}{N^4}\right) \\
    &=\frac{U^2}{\gamma}\left(\frac{16}{N}-\frac{96}{N^2}+\frac{274}{N^3}-\frac{386}{N^4}+\frac{192}{N^5}\right)
\end{split}
\end{equation}
This is the analytical value of the plateau which we compare with the exact numerics in the main text. As noted there, in the thermodynamic limit this value evidently scales asymptotically as $R_d(0)\sim(4U)^2/N$.

\subsection{Eigenbasis rate equations}
As noted in the main text, the emission dynamics resulting in the emission peak in this regime can be captured by rate equations in the eigenbasis of $H_d$. To keep the notation light, we denote arbitrary eigenstates of $H_d$ by $\ket{\lambda}$. Since $H_d$ conserves the total number of excitations, the eigenstates $\ket{\lambda}$ can be labelled by both the eigenenergy $E_\lambda$ and a fixed number of excitations $N_\lambda$ (up to some potential degeneracy). We insert the diagonal form $\rho(t)=\sum_\lambda p(\lambda,t)\ket{\lambda}\bra{\lambda}$ into the effective master equation~\eqref{eq:effective_lindbladian} to obtain equations of motion for the probabilities $p(\lambda,t)=\bra{\lambda}\rho(t)\ket{\lambda}$, which read
\begin{equation}
\label{eq:eigenbasis_rate_equations}
\begin{split}
    \dot{p}(\lambda,t)
    &=\sum_{\alpha=1,2}\Gamma_\alpha\left(\bra{\lambda}c_\alpha\rho(t)c_\alpha^\dagger\ket{\lambda}-\frac{1}{2}\bra{\lambda}c_\alpha^\dagger c_\alpha\rho(t)\ket{\lambda}-\frac{1}{2}\bra{\lambda}\rho(t)c_\alpha^\dagger c_\alpha\ket{\lambda}\right) \\
    &=\sum_{\alpha=1,2}\Gamma_\alpha\left(\sum_{\lambda'}\abs{\bra{\lambda}c_\alpha\ket{\lambda'}}^2p(\lambda',t)-\left(\sum_{\lambda'}\abs{\bra{\lambda'}c_\alpha\ket{\lambda}}^2\right)p(\lambda,t)\right) \\
    &=\underbrace{\sum_{\lambda'}R_{\lambda'\to\lambda}\,p(\lambda',t)}_{\textstyle\text{gain}}-\underbrace{\left(\sum_{\lambda'}R_{\lambda\to\lambda'}\right)p(\lambda,t)}_{\textstyle\text{loss}}
\end{split}
\end{equation}
Here, we have defined the transition rate $R_{\lambda\to\lambda'}=\sum_{\alpha=1,2}\Gamma_\alpha\abs{\bra{\lambda'}c_\alpha\ket{\lambda}}^2$ from state $\ket{\lambda}$ to state $\ket{\lambda'}$. The dark-state emission rate $R_d(t)=\sum_\alpha\Gamma_\alpha\mean{c_\alpha^\dagger c_\alpha}_t$ can then be computed as
\begin{equation*}
    R_d(t)=\sum_{\alpha,\lambda}\Gamma_\alpha\bra{\lambda}c_\alpha^\dagger c_\alpha\ket{\lambda}\,p(\lambda,t)=\sum_{\alpha,\lambda,\lambda'}\Gamma_\alpha \abs{\bra{\lambda'}c_\alpha\ket{\lambda}}^2\,p(\lambda,t)=\sum_{\lambda,\lambda'}R_{\lambda\to\lambda'}\,p(\lambda,t)\,.
\end{equation*}
In Fig.~\ref{fig:smallU_rate_equations}, we show a representative example of the emission rate dynamics predicted by the rate equations~\eqref{eq:smallU_rate_equations} and snapshots of the associated probability distribution dynamics. We find that the rate equations reproduce the dynamics of the dark-state emission rate $R_d(t)$ with near-perfect accuracy. 

Notably, we also find that the dynamics evidently remain confined entirely to the ground state manifold of $H_d$, spanned by the lowest-energy eigenstates in each excitation sector. This observation motivates us to restrict the rate equations~\eqref{eq:smallU_rate_equations} explicitly to the ground state manifold. We denote the (unique) $n$-excitation ground state of $H_d$ by $\ket{\psi_n}$ and the associated excitation probability $p_n(t)=\bra{\psi_n}\rho(t)\ket{\psi_n}$. Noting that $c_1$ and $c_2$ remove exactly a single excitation and two excitations, respectively, we then define the transition rates $R_n^{(1)}=\Gamma_1\abs{\bra{\psi_{n-1}}c_1\ket{\psi_n}}^2$ and $R_n^{(2)}=\Gamma_2\abs{\bra{\psi_{n-2}}c_2\ket{\psi_n}}^2$, such that the explicit rate equations in the ground state manifold take the form
\begin{equation}
\label{eq:ground_state_rate_equations}
    \dot{p}_n(t)\approx R^{(1)}_{n+1}p_{n+1}(t)+R^{(2)}_{n+2}p_{n+2}(t)-\left[R_n^{(1)}+R_n^{(2)}\right]p_n(t)\,.
\end{equation}
Consistent with our choice of initial state~\eqref{eq:initial_state_dark}, we choose the initial probability distribution $p_n(0)=\abs{C(n)}^2$, defined in Eq.~\eqref{eq:initial_state_dark_distribution}. In Fig.~\ref{fig:smallU_rate_equations}, we find that this explicit restriction to the ground state manifold indeed retains the agreement with the full simulation of the effective master equation. In the main text, we note that $\Gamma_2/\Gamma_1\sim N^{-1}$ and hence the terms associated with transitions under $c_2$ become negligible for the emission peak in the thermodynamic limit. Note, however, that we find that these terms do impact the slow late-time dynamics (see Fig.~\ref{fig:smallU_rate_equations}) in a non-negligble manner.

\begin{figure}
    \centering
    \includegraphics[width=\linewidth]{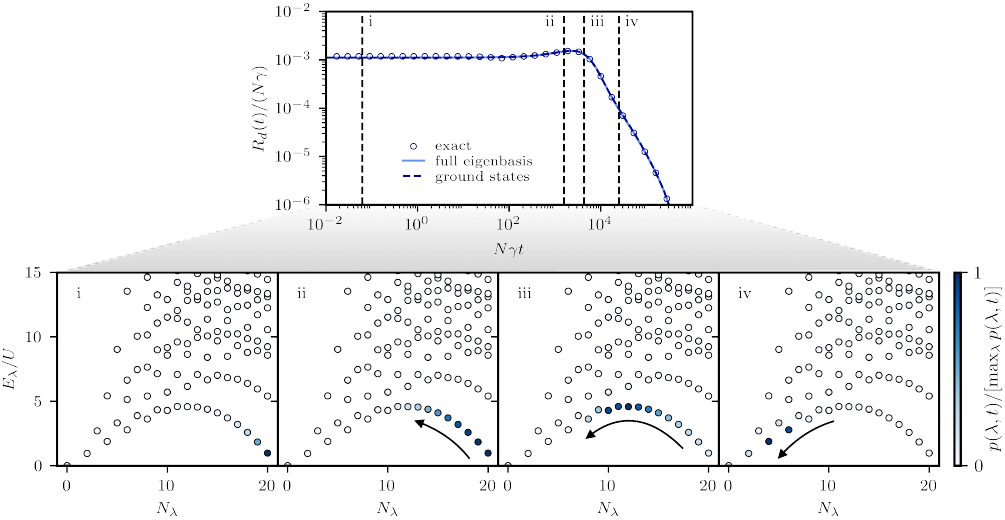}
    \caption{\emph{Small-$U$ rate equation dynamics.} Top panel: Comparison of the exact dark-state emission rate dynamics under the effective master equation~\eqref{eq:effective_lindbladian} (markers) with the rate equations~\eqref{eq:eigenbasis_rate_equations} in the full eigenbasis of the effective Hamiltonian~\eqref{eq:effective_hamiltonian} (solid line) and restricted to the ground state manifold (dashed line), for $N=20$ and $U=0.2\gamma$. Bottom panels: Snapshots of the probability distribution $p(\lambda,t)$ in the eigenbasis, with eigenstates arranged by the associated eigenenergy $E_\lambda$ and excitation number $N_\lambda$. The associated times are indicated by vertical dashed lines in the top panel.}
    \label{fig:smallU_rate_equations}
\end{figure}

\section{Exploring the eigenstate structure of $H_d$}

\subsection{Structure of the dark subspace}
In Sec.~\ref{app:symmetric_space} we have constructed a basis for the fully symmetric Hilbert space $\mathcal{H}$ where each state is identified by an integer partition (see Fig.~\ref{fig:truncation}). While this basis diagonalizes the Hamiltonian $H$ and enables a straightforward physically-motivated truncation of the Hilbert space, the construction of states spanning the dark subspace $\mathcal{H}_d$ of $\mathcal{H}$ (i.e. the space of states $\ket{\psi}\in\mathcal{H}$ obeying $c\ket{\psi}=0$) is non-trivial in this basis. Here, we introduce an alternative algebraic construction which simultaneously provides a basis for both $\mathcal{H}$ and $\mathcal{H}_d$ when $N=K=M$. While less explicit, this new basis offers an intuitive picture of the structure of $\mathcal{H}_d$ and an efficient way to compute its dimension. Note that an analogous algebraic construction was previously devised, in a different setting, in Ref.~\cite{Lewenstein_JPhysB_2000}.

The starting point is constructing a basis for $\mathcal{H}_d$. To do so we define, for $k>0$, the operators
\begin{equation*}
    D_k=\begin{cases}
        a_k a_{-k}\,, & k\ne \pi \\
        \dfrac{a_\pi^2}{2}\,, & k= \pi \\
    \end{cases} \quad\text{and}\quad
    D^0_k=\begin{cases}
        \dfrac{a^\dagger_k a_k + a_{-k} a_{-k}^\dagger}{2} \,, & k\ne \pi \\
        \dfrac{a^\dagger_\pi a_\pi + a_\pi a_\pi^\dagger}{4}  \,, & k= \pi \\
    \end{cases}
\end{equation*}
and observe that they generate independent copies of a $su(1,1)$ algebra $[D_k,D_{k'}^3] = \delta_{kk'} D_k$ and $[D_k,D_{k'}^\dagger] = \delta_{kk'} 2D^3_k$. From these operators, we then construct the collective generators
\begin{equation*}
    D=\sum_{k>0} D_k \quad\text{and}\quad D^3=\sum_{k>0} D_k^3
\end{equation*}
that also satisfy a $su(1,1)$ algebra $[D,D^\dagger]=2D^3$, $[D,D^3]=D$ and admit a simple physical interpretation: while $2D^3$ counts the number of excitations in the dark modes, $D=c_2/2$ ($D^\dagger=c^\dagger_2/2$) removes (adds) pairs of them.

Taking advantage of these collective generators we can iteratively build a basis for $\mathcal{H}_d$. As depicted in the right panel of Fig.~\ref{fig:dark_subspace}, the dark subspace is spanned by multiple towers of states generated by a repeated application of $D^\dagger \sim c_2^\dagger$. In fact, from the $su(1,1)$ algebra it follows that $D^\dagger D$ and $D^3$ are hermitian commuting operators and if $\ket{\psi}$ is a common eigenstate then $D^\dagger \ket{\psi}$ is also a common eigenstate with 2 more excitations. Hence, the only non-trivial step needed to explicitly construct these ladders is to identify their corresponding set of orthogonal, lowest-weight states, namely those dark states $\ket{\psi_d}$ that are both eigenstates of $D^3$ and satisfy $D\ket{\psi_d}=0$. An obvious lowest-weight dark state is the vacuum $\ket{0}$, however there exist also orthogonal, lowest-weight dark states in the finite excitation sectors. In Fig.~\ref{fig:dark_subspace}, we schematically depict both $\mathcal{H}$ and $\mathcal{H}_d$ for $N=K=M=6$, denoting the lowest-weight dark states as $\ket{3}$, $\ket{4}$, $\ket{5}$, $\ket{6.1}$, and $\ket{6.2}$, based on their excitation number and the additional multiplicity that starts to appear in the $6$-excitation sector. Once these orthogonal lowest-weight dark states are identified for a given $N$, an orthogonal basis for $\mathcal{H}_d$ is straightforwardly generated by the action of $D^\dagger \sim c_2^\dagger$ (see right panel in Fig.~\ref{fig:dark_subspace}). Moreover, as shown in the left panel of Fig.~\ref{fig:dark_subspace}, we can also straightforwardly generate an orthogonal basis for the whole Hilbert space $\mathcal{H}$ by repeatedly acting with $c^\dagger \propto a_0^\dagger$ on the different dark basis states (noting that $[c,c_2]=0$), thus generating states with a finite number of excitations in the bright mode.

Even though there is no simple way to expand the lowest-weight dark states and their corresponding ladders in the basis defined in Sec.~\ref{app:symmetric_space}, the only information we need for the present construction is the number $g(n)$ of orthogonal, lowest-weight dark states in the $n$-excitation sector. As we now show, this number can be explicitly derived together with the dimension of $\mathcal{H}_{d}^n$~\textemdash~the dark subspace with exactly $n$ excitations~\textemdash~and the dimension of the whole dark subspace $\mathcal{H}_{d}$. To do so, we denote the number of integer partitions of $n$ by $p(n)$ and we recall that the dimension of the \emph{full} permutation-invariant Hilbert space with exactly $n$ excitations is given by $\text{dim}(\mathcal{H}^n)=p(n)$ (see Section~\ref{app:symmetric_space}). Then, from the ladder structure we have just described, it follows that $f(n)\equiv\text{dim}(\mathcal{H}_{d}^n)=\text{dim}(\mathcal{H}^n)-\text{dim}(\mathcal{H}^{n-1})=p(n)-p(n-1)$. Similarly, it follows that $\text{dim}(\mathcal{H}_{d})=\text{dim}(\mathcal{H}^N)=p(N)$ and that $g(n)$ is implicitly given by $f(n)=\sum_{s=0}^{\lfloor n/2\rfloor}g(n-2s)$, where $\lfloor x \rfloor$ denotes to floor function. From this recurring relation, we finally obtain $g(n)=[p(n)-p(n-1)]-[p(n-2)-p(n-3)]$, with the convention $p(0)=1$ and $p(n<0)=0$.

\begin{figure}
    \centering
    \includegraphics[width=\linewidth]{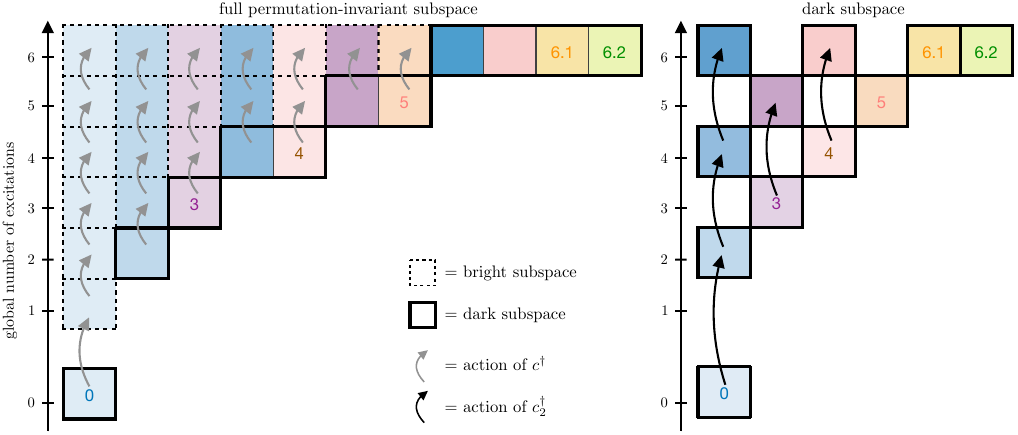}
    \caption{\emph{Algebraic structure of the dark subspace.} Diagrammatic representation of the construction (for $N=6$ modes with $K=M=6$) of the full symmetric Hilbert space $\mathcal{H}$ (left) and the dark subspace $\mathcal{H}_d$ (right) through the repeated action of $c^\dagger \sim a_0^\dagger$ and $c_2^\dagger \sim D^\dagger$ on dark lowest-weight states labeled by their excitation number and their multiplicity. The towers of dark states obtained through the repeated action of $c_2^\dagger$ on some lowest-weight state are identified by a color with increasing intensity.}
    \label{fig:dark_subspace}
\end{figure}

\subsection{Analytical approximate ladder states}
Despite the illustrative algebraic structure of the dark subspace (see Fig.~\ref{fig:dark_subspace}), analytically constructing the exact eigenstates of the effective Hamiltonian $H_d$ is not straightforwardly possible. The fact that the initial probability distribution in the eigenbasis is confined to the ground state manifold (see Fig.~\ref{fig:smallU_rate_equations}) might suggest that we could straightforwardly identify the $n$-excitation ground state $\ket{\psi_n}$ with the state $\ket{\phi_n}$ as defined in Eq.~\eqref{eq:phi_and_J_states}, which appears in the initial state~\eqref{eq:initial_state_dark}. While we find that this correspondence is \emph{not} exact, the states $\ket{\phi_n}$ do allow for an analytical computation of the transition rates $R_n^{(1)}$ and $R_n^{(2)}$. We therefore nonetheless include here an analysis of these states, which serve as an example of an analytically tractable approximate family of ladder states, capturing many of the core features of the ground states in the context of the small-$U$ emission dynamics.

\subsubsection{Action of $\tilde{a}_i$, $\tilde{a}_i^2$ and $c_2$ on $\ket{\phi_n}$}
We first investigate the action of operators $\tilde{a}_i$, $\tilde{a}_i^2$ on the states $\ket{J}$ defined in Eq.~\eqref{eq:phi_and_J_states} and, in turn, on the states $\ket{\phi_n}$. Recalling that $[\tilde{a}_i,\tilde{a}_j^\dagger]=\delta_{ij}-1/N$, we can commute the operators $\tilde{a}_j$ past the the $\tilde{a}_j^\dagger$ with $j\in J$ which appear in the definition of $\ket{J}$, to obtain
\begin{align*}
    &\tilde{a}_i\ket{J}=\delta_{i\in J}\ket{J\backslash\{i\}}-\frac{1}{N}\sum_{k\in J}\ket{J\backslash\{k\}} \\
    &\tilde{a}_i^2\ket{J}=-\frac{2}{N}\delta_{i\in J}\sum_{k\in J}\ket{J\backslash\{i,k\}}+\frac{1}{N^2}\sum_{k\in J}\sum_{q\in J\backslash\{k\}}\ket{J\backslash\{k,q\}}
\end{align*}
The action of $\tilde{a}_i$ and $\tilde{a}_i^2$ on $\ket{\phi_n}$ can then be obtained by summing these expressions over all sets $J$ of size $\abs{J}=n$,
\begin{equation*}
    \tilde{a}_i\ket{\phi_n}=\underbrace{\frac{1}{f(n)}\sum_{\substack{\abs{J}=n-1 \\ i\notin J}}\ket{J}}_{\textstyle \ket{\phi_n^{(i)}}}-\frac{1}{Nf(n)}\sum_{\abs{J}=n}\sum_{k\in J}\ket{J\backslash\{k\}}
\end{equation*}
Let's consider some fixed $J'$ with $|J'| = n - 1$. 
Trivially, this $J'$ is the state which appears in the sum over $k$ in the above expression whenever $J = J' \cup \{k\}$. There are $N - (n - 1)$ ways to pick a $k \in \{1,2,\ldots,N\} \setminus J'$, such that
\begin{equation*}
    \sum_{\abs{J}=n}\sum_{k\in J}\ket{J\backslash\{k\}}=\sum_{\abs{J}=n-1}\sum_{k\notin J}\ket{J}=(N-n+1)\sum_{\abs{J}=n-1}\ket{J}\,.
\end{equation*}
By a similar combinatorial argument, we can show that $\tilde{a}_i\ket{\phi_n^{(i)}}=-\alpha(n)\ket{\phi_{n-1}^{(i)}}$, where $\alpha(n)=g(n)f(n-1)/f(n)$ in terms of $g(n)=(N-n+1)/N$. Accordingly, the action of $\tilde{a}_i$ and $\tilde{a}_i^2$ on $\ket{\phi_n}$ can therefore be expressed as
\begin{subequations}
\label{eq:ai_ai2_phi}
\begin{align}
    &\tilde{a}_i\ket{\phi_n}=\ket{\phi_n^{(i)}}-\alpha(n)\ket{\phi_{n-1}} \label{eq:ai_phi} \\
    &\tilde{a}_i^2\ket{\phi_n}=\alpha(n)\alpha(n-1)\ket{\phi_{n-2}}-2\alpha(n)\ket{\phi^{(i)}_{n-1}} \label{eq:ai2_phi}
\end{align}
\end{subequations}
Note that our expressions strictly only hold for $n > 2$. Recalling that $\ket{\phi_1}$ does not exist, for the remaining cases $n=0,2$ we separately find that $\tilde{a}_i\ket{\phi_0}=\tilde{a}_i^2\ket{\phi_0}=0$ trivially (since $\ket{\phi_0}=\ket{\mathrm{vac}}$), while $\tilde{a}_i\ket{\phi_2}=-\ket{\{i\}}/f(2)$ and $\tilde{a}_i^2\ket{\phi_2}=-g(2)/f(2)\ket{\phi_0}$.

To compute the action of $c_2$ on $\ket{\phi_n}$ from Eqs.~\eqref{eq:ai_ai2_phi}, we first recall that in terms of the operators $\tilde{a}_j$, the effective jump operators $c_1=\sum_j\tilde{a}_j^\dagger\tilde{a}_j^2$ and $c_2=\sum_j\tilde{a}_j^2$. Then, noting that $\sum_i\ket{\phi^{(i)}_n}=N\alpha(n)\ket{\phi_{n-1}}$, it follows directly from Eq.~\eqref{eq:ai2_phi} that $c_2\ket{\phi_n}=-N\alpha(n)\alpha(n-1)$.

\begin{figure*}
    \centering
    \includegraphics[width=\linewidth]{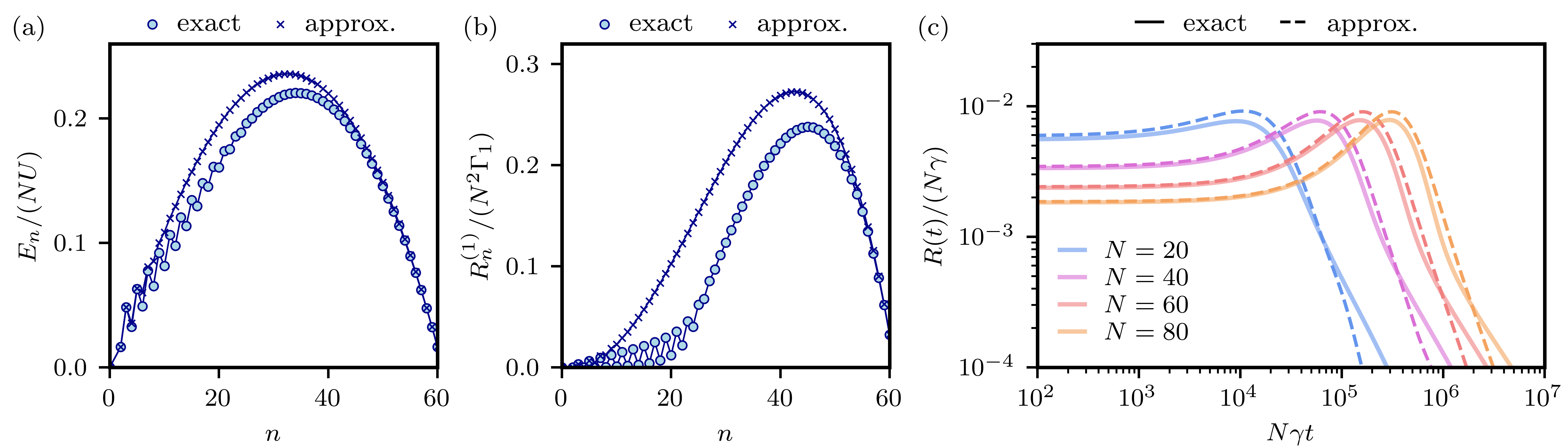}
    \caption{\emph{Approximate ladder ground states.} (a) Ground state energies $E_n$ and (b) transition rates $R_n^{(1)}$ in the ground state manifold for $N=60$. We plot both the exact quantities computed from the numerically obtained ground states of the effective Hamiltonian~\eqref{eq:effective_hamiltonian} (circles), as well as the analytical expressions from Eq.~\eqref{eq:approx_E} and Eq.~\eqref{eq:approx_R1} for the approximate ladder states (crosses). (c) Comparison of the dark-state emission rate dynamics for $U=0.1\gamma$ and varied system sizes $N$, obtained from the rate equations~\eqref{eq:ground_state_rate_equations} for the exact numerical ground states (solid lines) and the analytical ladder states (dashed lines).}
    \label{fig:approx_rate_equations}
\end{figure*}

\subsubsection{Energy and transition rates}
We now want to compute the energy expectation value $E_n=\bra{\phi_n}H_d\ket{\phi_n}$ associated with the effective Hamiltonian~\eqref{eq:effective_hamiltonian}, as well as the transition rates $R_n^{(1)}=\Gamma_1\abs{\bra{\phi_{n-1}}c_1\ket{\phi_n}}^2$ and $R_n^{(2)}=\Gamma_2\abs{\bra{\phi_{n-2}}c_2\ket{\phi_n}}^2$. Using the expression for $c_2\ket{\phi_n}$ obtained above, we immediately find that
\begin{equation}
\label{eq:approx_R2}
    R_n^{(2)}=\begin{cases}
        N^2\alpha^2(n)\alpha^2(n-1)\Gamma_2\,, & n>2 \\
        2(N-1)\Gamma_2, & n=2 \\
        0\,, & \text{otherwise}
    \end{cases}
\end{equation}
We now turn to $E_n$ and $R_n^{(1)}$. We compute these as the norm of $\tilde{a}_i^2\ket{\phi_n}$ and the overlap between $\tilde{a}_i\ket{\phi_{n-1}}$ and $\tilde{a}_i^2\ket{\phi_n}$, summed over $i$, respectively. From Eqs.~\eqref{eq:ai_ai2_phi}, it then follows that $E_n$ and $R_n^{(1)}$ can be expressed in terms of the overlaps $\langle\phi_n|\phi_n^{(i)}\rangle$ and $\langle\phi_n^{(i)}|\phi_n^{(i)}\rangle$. To compute these overlaps, we recall that $\sum_j\tilde{a}_j=0$ and further note that $(\sum_j\tilde{a}_j^\dagger\tilde{a}_j)\ket{J}=\abs{J}\ket{J}$. Accordingly, we find $\sum_i\langle\phi_n|\phi_n^{(i)}\rangle = N\alpha(n)$ and $\sum_i\langle\phi_n^{(i)}|\phi_n^{(i)}\rangle=n+N\alpha^2(n)$. Hence,
\begin{equation}
\label{eq:approx_R1}
    R_n^{(1)}=\begin{cases}
        4\alpha^2(n)(n-1)^2\Gamma_1\,, & n\geq 2 \\
        0\,, & \text{otherwise}
    \end{cases}
\end{equation}
Likewise, we also obtain an expression for the energy expectation values $E_n$, which reads
\begin{equation}
\label{eq:approx_E}
    E_n=\begin{cases}
        U\alpha^2(n)\left[4(n-1) + N\alpha^2(n-1)\right]/2\,, & n>2 \\
        U(N-1)/N\,, & n=2 \\
        0\,, & \text{otherwise}
    \end{cases}
\end{equation}

\subsubsection{Agreement with exact ground states}
We find that $E_n$ shows good qualitative agreement with the true ground state energies (see Fig.~\ref{fig:approx_rate_equations}a). Quantitatively, we observe good agreement at high and low excitation numbers, however in the ``bulk'' of the spectrum we observe significant quantitative disagreement which persists in the thermodynamic limit. We find that the same holds for the transition rate $R_n^{(1)}$ (see Fig.~\ref{fig:approx_rate_equations}b), where we further see that the analytical expression~\eqref{eq:approx_R1} also reproduces the scaling $\sim N^2\Gamma_1$ of the peak transition rate. For the other set of transition rates $R_n^{(2)}$, we see significantly worse quantitative agreement with the (numerically obtained) transition rates for the true ground states. Nonetheless, the dynamics under the rate equations~\eqref{eq:ground_state_rate_equations} under the substitution of Eqs.~\eqref{eq:approx_R2} and~\eqref{eq:approx_R1} do yield quite good agreement with the exact dynamics under~\eqref{eq:effective_lindbladian} (see Fig.~\ref{fig:approx_rate_equations}).

\subsection{Details on the de Finetti calculations}

\subsubsection{Role of the local truncation $M$}
As discussed in the End Matter, in order to apply the quantum de Finetti theorem to our system, we have to introduce a local excitation cutoff $M$ (see Section~\ref{app:symmetric_space}). The choice of $M$ has a non-trivial effect on the obtained ground-state energies $E_\nu$ and transition rates $R_\nu^{(1)},R_\nu^{(2)}$ (see Fig.~\ref{fig:de_finetti_truncation}): For $M=2$, we obtain a different critical excitation density ($\nu_c=1/2$) than for $M>2$ ($\nu_c=1-1/\sqrt{2}$). This underlines quite clearly that we have to account for multiple occupation beyond doublons in order to accurately capture the key quantities characterising the dynamics in the weak-interaction limit. For $M>2$ (which is implicitly assumed in the proofs shown in the next section), we find that the dependence on $M$ is less dramatic, and in particular our results essentially converge for $M=4$. This is consistent with the choice of $M=4$ in our numerical simulations, which, as noted in Section~\eqref{app:symmetric_space}, accurately captures the dynamics for any $U$ and $N$.

\begin{figure*}
    \centering
    \includegraphics[width=\linewidth]{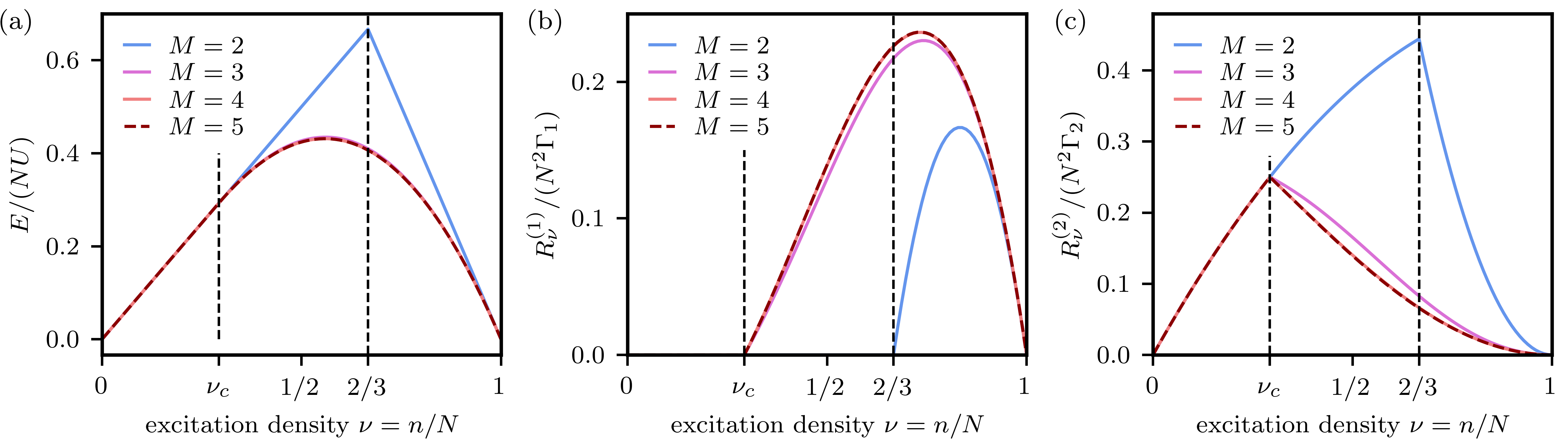}
    \caption{\emph{De Finetti energies and transition rates.} (a) Ground state energies and (b),(c) transition rates in the ground state manifold, obtained using the de Finetti Theorem under various local excitation cutoffs $M$.}
    \label{fig:de_finetti_truncation}
\end{figure*}

\subsubsection{Optimality of the doublon single-mode states}

As noted in the End Matter of the main text, a family of single-particle states $\ket{\varphi}$ which obeys the required constraints $\bra{\varphi}a\ket{\varphi}=0$ and $\bra{\varphi}a^\dagger a\ket{\varphi}=\nu$ at any $\nu$ is given by
\begin{equation}
    \ket{\phi_\nu^\mathrm{ev}}=\sqrt{1-\frac{\nu}{2}}\ket{0}+\sqrt{\frac{\nu}{2}}\ket{2}
\end{equation}
These states have a single-particle energy $\mathcal{E}(\phi)=\nu$, which we observe numerically to be minimal at $\nu <\nu_c$, with the critical excitation density given by $\nu_c=1-1/\sqrt{2}$. Below, we prove this explicitly. Subsequently, we also show that, contrary to $\ket{\phi_\nu^\mathrm{ev}}$, the states which minimize the single-particle energy at $\nu>\nu_c$ have non-zero amplitudes on Fock states with odd excitation numbers.

\begin{proposition}[Optimality]
For excitation densities $\nu$ up to $\nu_c= 1-1/\sqrt{2}$, $\ket{\phi_\nu^\mathrm{ev}}$ attains the global energy minimum.
\end{proposition}

\begin{proof}
Let $\ket{\phi}=\sum_n c_m\ket{m}$ (explicit parametrization of $\ket{\phi}$ in the single-particle Fock basis), $p_m=|c_m|^2$, and $w=\sum_{m\geq 3}mp_m$, and suppose $\mathcal{E}(\phi)<\nu$. We will show that this leads to a contradiction. First, consider
\begin{equation}
    \Delta\coloneqq\nu-\mathcal{E}(\phi) = p_1-\sum_{m\geq 3}m(m-2)p_m>0.
\end{equation}
Since $m(m-2)\geq m$ for all $m\geq 3$, the sum is at least $w$, and thus $0<\Delta\leq p_1-w$, which implies that $p_1>w>0$. Next, the dark-state condition $0=\sum_{m\geq 0}\sqrt{m+1}\,c_m^*c_{m+1}$ and the triangle inequality together yield
\begin{equation}
	\begin{aligned}
		\sqrt{p_0p_1}&=|c_0^*c_1|\leq \sqrt{2p_1p_2}+T,
	\end{aligned}
\end{equation}
where the tail $T\coloneqq\sum_{m\geq 2}\sqrt{(m+1)p_mp_{m+1}}$.
Since $2p_2=\nu-p_1-w\leq\nu-p_1$, the first term satisfies $\sqrt{2p_1p_2}\leq\sqrt{p_1(\nu-p_1)}$.
For the tail, Cauchy--Schwarz gives
\begin{equation}
    T^2\leq\Bigl(\sum_{n\geq 2}np_n\Bigr)\Bigl(\sum_{m\geq 3}\frac{m}{m-1}p_m\Bigr)=(\nu-p_1)\sum_{m\geq 3}\frac{m}{m-1}p_m.
\end{equation}
Since $m-1\geq m/2$ for $m\geq 2$, we have $m/(m-1)\leq m/2$, so $\sum_{m\geq 3}mp_m/(m-1)\leq w/2<p_1/2$, giving $T<\sqrt{p_1(\nu-p_1)/2}$. Combining and dividing by $\sqrt{p_1}>0$ we finally arrive at an upper bound on $p_0$, which reads
\begin{equation}
    p_0<(\nu-p_1)\left(1+\frac{1}{\sqrt{2}}\right)^{2}.
\end{equation}
In the final step, we combine this with a lower bound on $p_0$ from the relation $1=\sum_m p_m\leq p_0+(p_1+\nu)/2$,
\begin{equation}
	1-\frac12(p_1+\nu)\leq\frac12(\nu-p_1)(1+\sqrt2)^2,
	\label{eq:pre-contradiction}
\end{equation}
which simplifies to $1\leq\nu(2+\sqrt2)-p_1(1+\sqrt2)<\nu(2+\sqrt2)=\nu/\nu_c$, which is a contradiction for $\nu\leq\nu_c$.
\end{proof}

\begin{proposition}[Critical density]
	For $\nu>\nu_c$, states with odd Fock components achieve lower energy than $\ket{\phi_\nu^{\mathrm{ev}}}$.
\end{proposition}
\begin{proof}
We perturb the even state by adding small odd components. To second order, we obtain
\begin{equation}
    \ket{\phi} = (A+\alpha\epsilon^2)\ket{0}+\epsilon b\ket{1}+(C+\gamma\epsilon^2)\ket{2}+\epsilon d\ket{3},
\end{equation}
with $b,d,\alpha,\gamma\in\mathbb{R}$, $A=\sqrt{1-\nu/2}$, and $C=-\sqrt{\nu/2}$.
The normalisation and excitation density constraints give two linear equations (at second order in $\epsilon$), which read
\begin{align}
    \text{normalisation:}&\quad 2(A\alpha+C\gamma) = -(b^2+d^2),\\
    \text{density:}&\quad 4C\gamma = -(b^2+3d^2).
	\label{eq:second-order-constraints}
\end{align}
These uniquely determine $\gamma=-(b^2+3d^2)/(4C)$ and $\alpha=(-b^2+d^2)/(4A)$.
The dark-state condition $0=\bra{\phi}a\ket{\phi}$ yields a constraint at first order in $\epsilon$
\begin{equation}
    (A+\sqrt{2}\,C)\,b + \sqrt{3}\,C\,d = 0
    \implies
	\frac{d}{b} = \frac{\sqrt{2/\nu-1}-\sqrt{2}}{\sqrt3}.
\end{equation}
We now compute the change in energy (at second order in $\epsilon$).
Using the density constraint \eqref{eq:second-order-constraints}, we find
\begin{equation}
    \mathcal{E} = \nu + [4C\gamma + 6d^2]\epsilon^2 = \nu + \epsilon^2(3d^2-b^2).
\end{equation}
The energy decreases if $(d/b)^2<1/3$, i.e. if $(\sqrt{2/\nu-1}-\sqrt{2})^2<1$, which is fulfilled for $\nu\in(\nu_c,1]$.
\end{proof}

\section{Large-$U$ rate equations \& multi-peak dynamics} 
\label{app:largeU_rate_equations}

\subsection{Rate equations in the symmetric basis}
To understand the large-$U$ dynamics of the family of multi-excited initial states $\ket{\psi_0}=\prod_j(a_j^\dagger)^{M_0}|\mathrm{vac}\rangle$ introduced in the End Matter of the main text, we consider the dynamics in the limit $U\to\infty$, which corresponds with dropping the term $\delta\mathcal{D}$ entirely in the decomposition~\eqref{eq:largeU_lindbladian_decomposition}. As noted in Section~\ref{app:largeU_perturbation} and the main text, neither $\mathcal{V}$ nor $\mathcal{D}_\infty$ generate  coherences in the symmetric basis. Assuming that there are no coherences in the initial state, this implies that the density operator retains its diagonal form $\rho(t)=\sum_\vec{x}p(\vec{x},t)\ket{\vec{x}}\bra{\vec{x}}$ throughout the dynamics. It is easy to see that $\mathcal{V}\rho(t)=0$ in this case, therefore we only need to account for the evolution of the probabilities $p(\vec{x},t)=\bra{\vec{x}}\rho(t)\ket{\vec{x}}$ under $\mathcal{D}_\infty=\gamma\sum_n n\mathcal{D}[L_n]$. To this end, we note that the jump operators $L_n$ act on a general basis state $\ket{\vec{x}}$ as $L_n\ket{\vec{x}}=\sqrt{x_n(x_{n-1}+1)}\ket{T_n\vec{x}}$ and $L_n^\dagger\ket{\vec{x}}=\sqrt{x_{n-1}(x_n+1)}\ket{T_n^{-1}\vec{x}}$, where $T_n$ and its inverse $T_n^{-1}$ are defined by
\begin{equation*}
    (T_n\vec{x})_m=\begin{cases}
        x_m-1\,, & m=n \\
        x_{m-1}+1\,, & m=n-1 \\
        x_m\,, & \text{otherwise}
    \end{cases}\quad\leftrightarrow\quad(T_n^{-1}\vec{x})_m=\begin{cases}
        x_m+1\,, & m=n \\
        x_{m-1}-1\,, & m=n-1 \\
        x_m\,, & \text{otherwise}
    \end{cases}\,.
\end{equation*}
We then substitute the diagonal form of the density operator into the dissipative dynamics under $\mathcal{D}_\infty$ to obtain
\begin{equation}
\label{eq:largeU_rate_equations}
\begin{split}
    \dot{p}(\vec{x},t)
    &=\gamma\sum_nn\left(\bra{\vec{x}}L_n\rho(t)L_n^\dagger\ket{\vec{x}}-\frac{1}{2}\bra{\vec{x}}L_n^\dagger L_n\rho(t)\ket{\vec{x}}-\frac{1}{2}\bra{\vec{x}}\rho(t)L_n^\dagger L_n\ket{\vec{x}}\right) \\
    &=\gamma\sum_nn\bigg(x_{n-1}(x_n+1)\bra{T_n^{-1}\vec{x}}\rho(t)\ket{T_n^{-1}\vec{x}}-x_n(x_{n-1}+1)\bra{\vec{x}}\rho(t)\ket{\vec{x}}\bigg) \\
    &=\underbrace{\sum_nW_n(T_n^{-1}\vec{x})\,p(T_n^{-1}\vec{x},t)}_{\textstyle\text{gain}}-\underbrace{\left(\sum_nW_n(\vec{x})\right)}_{\textstyle\text{loss}}p(\vec{x},t)\,\,,
\end{split}
\end{equation}
where we have defined $W_n(\vec{x})$ as the transition rate from the state $\ket{\vec{x}}$ to $\ket{T_n\vec{x}}$ under the jump operator $L_n$ as
\begin{equation}
\label{eq:largeU_transition_rates}
    W_n(\vec{x})=\gamma n\,\bra{\vec{x}}L_n^\dagger L_n\ket{\vec{x}}=\gamma n\,x_n(x_{n-1}+1)\,.
\end{equation}
The collective emission rate $R(t)=\gamma\,\mathrm{tr}\{c^\dagger c\rho(t)\}$ is then given by
\begin{equation*}
    R(t)=\gamma\sum_n n\mean{L_n^\dagger L_n}_t=\gamma\sum_{n,\vec{x}}n\,\bra{\vec{x}}L_n^\dagger L_n\ket{\vec{x}}\,p(\vec{x},t)=\sum_{n,\vec{x}}W_n(\vec{x})\,p(\vec{x},t)\,.
\end{equation*}

\subsection{Multi-dimensional `grid' of intersecting Dicke ladders}
It is important to note that integrating the rate equations~\eqref{eq:largeU_rate_equations} does not require the construction of the full symmetric Hilbert space. The (much smaller) space of states participating in the dynamics under Eq.~\eqref{eq:largeU_rate_equations} can be directly obtained through a repeated application of the jump operators $L_n$ to the initial state. In fact, since $\mathcal{D}_\infty$ does not create higher multiple occupancies than those present in the initial state, for the canonical initial states $\ket{\psi_0}=\prod_j(a_j^\dagger)^{M_0}|\mathrm{vac}\rangle$ we only need to consider jump operator $L_n$ with $n\leq M_0$.

This procedure results in an $M_0$-dimensional `grid' of states, which we visualise for the case of $N=3$ and $M_0=2$ (i.e. two initial excitations on three resonators) in Fig.~\ref{fig:largeU_rate_equations}a. Constructing only the relevant states in this way not only significantly reduces the computational complexity of the problem, but allows us to draw a very natural connection between the general rate equations~\eqref{eq:largeU_rate_equations} and the well-known Dicke rate equations. 

\subsubsection{Initial state $M_0=1$}
In the case of $M_0=1$, the states involved in the dynamics are $\ket{\bm{x}}=\ket{N-x,x,0,0,\ldots}$ ($x=0,1,\ldots,N$), which are simply the Dicke states with $x$ excitations (see Section~\ref{app:largeU_perturbation}). In this case, there is only a single jump operator $L_1$ and the associated decay rates are given by $W_1(\vec{x})=\gamma x(N-x+1)$. By comparing with the relevant expressions in Section~\ref{app:largeU_perturbation}, it is clear that $W_1(\vec{x})=W_x$ are the transition rates which form the characteristic Dicke ladder~\cite{gross_superradiance_1982}, and accordingly Eq.~\eqref{eq:largeU_rate_equations} corresponds exactly with the standard Dicke rate equations~\eqref{eq:populations} for $M_0=1$.

\subsubsection{Initial states $M_0>1$}
A similar analysis can be applied to the more general case $M_0>1$. To make this clear, we define generalizations of the Dicke rates as $W_n^{(m)}=\gamma n(m-n+1)$, such that the standard Dicke ladder transition rates $W_n=W_n^{(N)}$. Then, in the case of $M_0=2$, for example, the relevant subspace is spanned by states $\ket{\vec{x}}\equiv\ket{N-x_1-x_2,x_1,x_2,0,\ldots}$ (with $x_1,x_2=0,1,\ldots,N$ and $x_1+x_2\leq N$) and we have to consider two jump operators $L_1,L_2$. The action of $L_1$ leaves $x_2$ invariant, so at some fixed $x_2$ (vertical lines in Fig.~\ref{fig:largeU_rate_equations}a) we find that $W_1(\vec{x})=\gamma x_1(N-x_2-x_1+1)=W_{x_1}^{(N-x_2)}$, corresponding with a local Dicke ladder in $x_1$ of length $N-x_2$. Similarly, $L_2$ leaves $x_1+x_2$ invariant, and hence at some fixed $x_1+x_2$ (horizontal lines in Fig.~\ref{fig:largeU_rate_equations}a) we find that $W_2(\vec{x})=2\gamma x_2(x_1+1)=2W_{x_2}^{(x_1+x_2)}$, corresponding with a local Dicke ladder in $x_2$ of length $x_1+x_2$, enhanced by an overall factor of $2$. 

This extends straightforwardly to general $M_0>1$, and we can concisely relate the general transition rates $W_n(\vec{x})$ to the (generalized) Dicke ladder transition rates $W_n^{(m)}$ as
\begin{equation}
    W_1(\vec{x})=W_{x_1}^{(N-x_2-x_3-\ldots)} \quad\text{and}\quad W_{k>1}(\vec{x})=k\,W_{x_k}^{(x_k+x_{k-1})}\,,
\end{equation}
demonstrating that the states participating in the rate equation dynamics can be viewed as a grid of intersecting Dicke ladders associated with each of the decay channels $L_n$ ($n=1,\ldots,M_0$) enhanced by a factor of $n$, respectively. This observation also allows us to understand the multi-peak structure that emerges for sufficiently large $N$ and $M_0>1$, as noted in the End Matter of the main text; In Fig.~\ref{fig:largeU_rate_equations}b we present an example of the two-peak emission dynamics for $M_0=2$, and in Fig.~\ref{fig:largeU_rate_equations}c we visualise the corresponding probability distribution dynamics. Clearly, the dynamics can be viewed, at least approximately, as the sequence of an initial decay along the (two-fold enhanced) Dicke ladder associated with $L_2$, giving rise to the first burst, and a subsequent decay along the Dicke ladder associated with $L_1$, giving rise to the second burst.

\begin{figure}
    \centering
    \includegraphics[width=\linewidth]{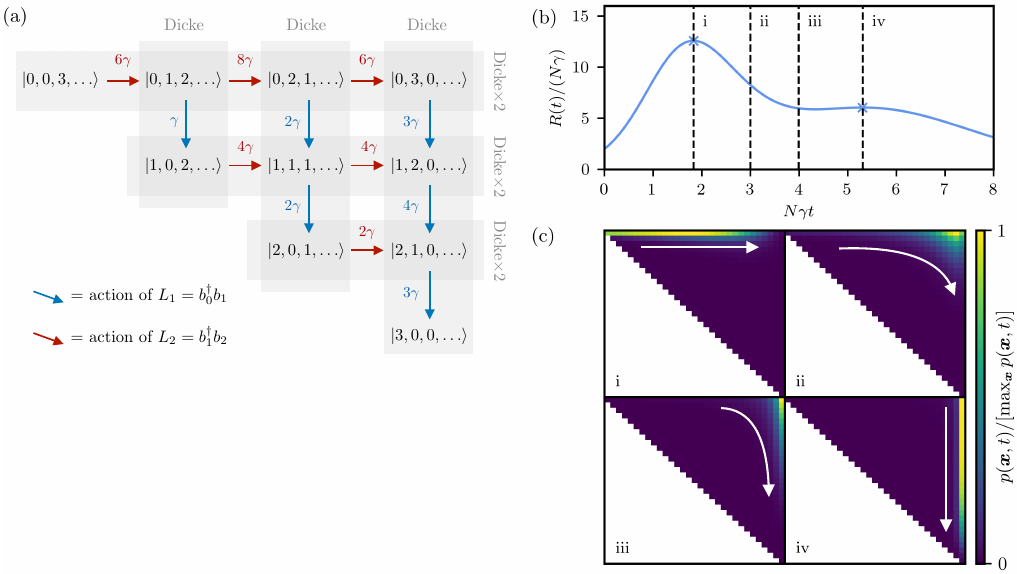}
    \caption{\emph{Large-$U$ rate equation dynamics.} (a) Visualisation of the symmetric basis states participating in the rate equations~\eqref{eq:largeU_rate_equations} for $N=3$ resonators and $M_0=2$. Transitions corresponding with the jumps $L_n$ ($n=1,2$) and the associated transition rates $W_n(\vec{x})$, as well as the interpretation in terms of intersecting Dicke ladders, are also indicated. (b) Emission dynamics for $N=30$, computed using the rate equations~\eqref{eq:largeU_rate_equations}. The collective emission rate $R(t)$ displays a two-peak structure, with the two emission peaks indicated by markers. (c) Snapshots of the (relative) instantaneous probability distribution $p(\vec{x},t)$ associated with the dynamics in (b), visualised on the relevant states arranged analogous to (a). The times corresponding with each of the panels are indicated by vertical lines in (b).}
    \label{fig:largeU_rate_equations}
\end{figure}

\section{Extracting the numerical peak emission rate scaling}
As noted in Section~\ref{app:symmetric_space}, we find that a local excitation cutoff $M=4$ (combined with a global excitation cutoff $K=N$) captures the dynamics of the initial state $\ket{\psi_0}=\prod_ja_j^\dagger\ket{\rm vac}$ accurately for all values of $U$. This allows us to perform large-scale simulations of the emission rate dynamics, based on the quantum trajectory method~\cite{dalibard_wave-function_1992,molmer_monte_1993} implemented through the \texttt{QuantumOptics.jl} module~\cite{kramer_quantumoptics_2018}.

The numerical evidence for the cross-over from sub-quadratic to quadratic scaling presented in the main text is obtained by simulating the emission rate dynamics for $U=a\gamma N^\alpha$ with system sizes up to $N=70$ and fitting the asymptotic scaling $R_\mathrm{pk}\sim\gamma N^\beta$. Notably, we find that the choice of the  prefactor $a$ in the scaling of $U$ proves crucial in circumventing finite-size effects, as we show in Fig.~\ref{fig:convergence}. Specifically, for $\alpha<1$ we find that a pre-factor $a<1$ is required to ensure convergence on the attainable system sizes. For $\alpha>1$, however, the choice $a<1$ leads to an unphysical fitted scaling exponent $\beta >2$ (noting that $R(t)=N\gamma\mean{a_0^\dagger a_0}\leq N^2\gamma$ for $K=N$ excitations in the system), while the prefactor $a=1$ shows good convergence.

\begin{figure}[h!]
    \centering
    \includegraphics[width=\linewidth]{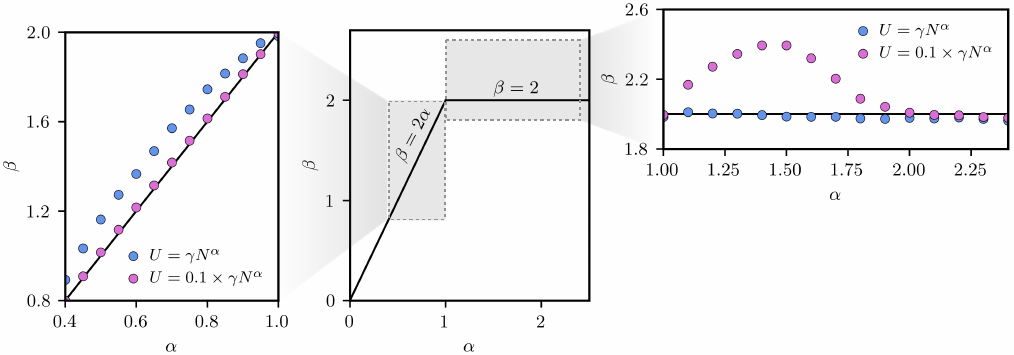}
    \caption{\emph{Fitted peak emission rate scaling.} Schematic illustration of the scaling crossover predicted theoretically in the main text (middle) and numerically fitted scaling exponents with different constant prefactors for different regimes (left/right).}
    \label{fig:convergence}
\end{figure}

\end{document}